\documentclass[a4paper,onecolumn,11pt,accepted=2021-11-12]{quantumarticle}
\pdfoutput=1
\usepackage{fontenc}
\usepackage[latin9]{inputenc}
\usepackage[english]{babel}
\usepackage{amsmath}
\usepackage{hyperref}
\usepackage{amsthm}
\usepackage{amssymb}
\usepackage{stmaryrd}
\usepackage{graphicx}
\usepackage{subfig}

\usepackage{tikz}
\usepackage{lipsum}

\theoremstyle{definition}
\newtheorem{defn}{\protect\definitionname}
\theoremstyle{plain}
\newtheorem{lem}{\protect\lemmaname}
\theoremstyle{plain}
\newtheorem{thm}{\protect\theoremname}
\theoremstyle{remark}
\newtheorem*{acknowledgement*}{\protect\acknowledgementname}

\usepackage{mathrsfs}
\usepackage{amsmath}
\usepackage{mathtools}
\usepackage{upgreek}
\usepackage{babel}
\usepackage{cite}
\usepackage{thmtools}
\usepackage{thm-restate}

\DeclareMathOperator{\Tr}{Tr}
\DeclareMathOperator{\I}{I}

\usepackage{enumitem}
\setenumerate{label=(\roman*)}

\newcommand{\C}{\mathbb{C}}
\newcommand{\R}{\mathbb{R}}
\newcommand{\cS}{\mathcal{S}}
\newcommand{\E}{\boldsymbol{e}}
\newcommand{\F}{\boldsymbol{f}}
\newcommand{\bu}{\boldsymbol{u}}
\newcommand{\ba}{\boldsymbol{a}}
\newcommand{\bx}{\boldsymbol{x}}
\newcommand{\bh}{\boldsymbol{h}}
\newcommand{\bO}{\boldsymbol{0}}
\newcommand{\bo}{\boldsymbol{\omega}}
\renewcommand{\sigma}{\upsigma}
\newcommand{\norm}[1]{\left\lVert#1\right\rVert}
\newcommand{\conv}{\mbox{Conv}}

\usepackage{babel}
\providecommand{\acknowledgementname}{Acknowledgement}
\providecommand{\definitionname}{Definition}
\providecommand{\lemmaname}{Lemma}
\providecommand{\theoremname}{Theorem}

\begin{document}

\title{General Probabilistic Theories with a Gleason-type Theorem}

\author{Victoria J Wright}
\affiliation{International Centre for Theory of Quantum Technologies, University of Gda\'{n}sk, 80-308 Gda\'{n}sk, Poland}
\affiliation{ICFO-Institut de Ciencies Fotoniques, The Barcelona Institute of Science and Technology, 08860 Castelldefels,
Spain}
\email{victoria.wright@icfo.eu}
\orcid{0000-0003-3523-7553}
\author{Stefan Weigert}
\email{stefan.weigert@york.ac.uk}
\orcid{0000-0002-6647-3252}
\affiliation{Department of Mathematics, University of York, York YO10 5DD, United Kingdom}

\maketitle

\begin{abstract}
Gleason-type theorems for quantum theory allow one to recover the
quantum state space by assuming that (i) states consistently assign
probabilities to measurement outcomes and that (ii) there is a unique
state for every such assignment. We identify the class of general
probabilistic theories which also admit Gleason-type theorems. It
contains theories satisfying the no-restriction hypothesis as well
as others which can simulate such an unrestricted theory arbitrarily
well when allowing for post-selection on measurement outcomes. Our
result also implies that the standard no-restriction hypothesis applied
to effects is not equivalent to the dual no-restriction hypothesis
applied to states which is found to be less restrictive.
\end{abstract}
\global\long\def\kb#1#2{|#1\rangle\langle#2|}%
\global\long\def\bk#1#2{\langle#1|#2\rangle}%

\global\long\def\ket#1{|#1\rangle}%
\global\long\def\bra#1{\langle#1|}%
\global\long\def\cd{\mathbb{C}^{d}}%

\section{Introduction}

More than sixty years ago, Mackey \cite{mackey1957} asked whether
the density operator represents the most general notion of a quantum
state that is consistent with the standard description of observables
as self-adjoint operators. Gleason \cite{Gleason1957} responded with
a proof that---in separable Hilbert spaces of dimension greater than
two---every state must admit an expression in terms of a density
operator if it is to consistently assign probabilities to the measurement
outcomes of such observables. In 2003, Busch \cite{Busch2003} (and
then Caves et al. \cite{Caves2004}) generalized the idea of Gleason's
theorem to observables represented by positive-operator measures (POMs).
The resulting Gleason\emph{-type }theorem (GTT) not only is much simpler
to prove but it also applies to two-dimensional Hilbert spaces, since
the assumptions being made are stronger than in Gleason's case.

In this paper, we investigate whether the Gleason-type theorem is
special to quantum theory. Imagine that a theory different from quantum
theory were to successfully describe Nature. Would a GTT still exist?

Our question is made explicit by posing it within the family of \emph{general
probabilistic theories} (GPTs) which have emerged as natural generalizations
of quantum theory \cite{janotta2014generalized,BarrettGPT,Masanes2011,barnum2011information,Hardy2001a}.
The framework of GPTs derives from operational principles and it encompasses
both quantum and classical models. One of the motivations to explore
these alternative theories has been to identify features which single
out quantum theory among others of comparable structure. Our study
contributes to that fundamental quest.

Effectively, Gleason and Busch establish a bijection between frame
functions and density operators in quantum theory. \emph{Frame functions
}associate probabilities to the mathematical objects representing
the possible outcomes of measurements in such a way that the probabilities
assigned to all disjoint outcomes of a given measurement sum to unity.
The rationale behind a frame function is that the probabilities of
all measurement outcomes for all observables should define a unique
state. If this were not the case, then two ``different'' states
would be indistinguishable, both practically and theoretically.

Our strategy will be to generalise the concept of frame functions
to GPTs in order to investigate whether they are in exact correspondence
with the objects that represent states in these theories. We are able
to identify all general probabilistic theories in which this correspondence
continues to hold. We find that GTTs exist for GPTs satisfying the
no-restriction hypothesis \cite{ChiribellaPhysRevA.81.062348,janottanorestriction},
as anyone familiar with the proof of Busch's result or the work of
Gudder et al. \cite{gudder1999convex} might expect. However, we also
find other GPTs which admit a Gleason-type theorem, namely those that
satisfy a ``noisy version'' of the no-restriction hypothesis (or come
arbitrarily close to satisfying it). An alternative way to characterise
this class of GPTs is to use the idea of simultation of measurements
via classical processes and postselection \cite{postselection}. Any
GPT in this class can simulate an arbitrarily good approximation to
any observable in a related unrestricted GPT (i.e. the GPT satisfying
the no-restriction hypothesis).

The existence of a GTT for GPTs such as quantum theory or \emph{real-vector-space
quantum theory }\cite{caves2001entanglement,Hardy2012,wootters2012entanglement,Aleksandrova2013}
has a number of consequences. It becomes possible, for example, to
modify the axiomatic structure of the theories as it is no longer
necessary to---separately and independently---sti\-pu\-late both
the state space \emph{and} the observables of the theory. Our result
can also be used to derive the standard GPT framework from operational
assumptions different to those found in the literature. More specifically,
the standard GPT framework is recovered if---after motivating the
standard description of observables in GPTs---states are assumed
to correspond to frame functions of these observables.

Additionally, our result also uncovers a new property of the no-restriction
hypothesis in GPTs. The standard no-restriction hypothesis applied
to effects turns out to be inequivalent to the dual assumption---the
\emph{no-state-restriction hypothesis---}applied to states. Consequently,
deriving the GPT framework in two equivalent ways, starting from the
structure of state spaces in one case or effect spaces in the other,
leads to inequivalent frameworks if the respective no-restriction
hypotheses are assumed.

To make the paper self-contained and to introduce the notation, we
will first review concepts of the GPT framework relevant here. In
Section \ref{sec:Main-Result}, we define frame functions for GPTs
and prove our main theorem.\emph{ }In Section \ref{sec:Examples},
we provide three examples to demonstrate the simplification of the
postulates required to specify an individual GPT. Section \ref{sec:Simulability}
strengthens our main theorem by defining frame functions only on a
proper subset of all observables, the analog of\emph{ projective-simulable
}observables. The stronger result leads to an alternative operational
motivation for deriving part of the GPT framework. In Section \ref{sec:Summary-and-Discussion}
we summarize and discuss our results.

\section{General probabilistic theories\label{sec:Generalized-probabilistic-theori}}

The GPT framework allows one to define a broad family of theories
of which quantum theory (in finite dimensional Hilbert spaces) is
a member. Any (real or fictitious) system described by a GPT has the
following fundamental property: there exists a finite set of \emph{fiducial
}measurement outcomes, the probabilities of which uniquely determine
its state.\footnote{To encompass quantum theory \emph{in toto}, the restriction to a \emph{finite}
set of fiducial measurement outcomes would need to be relaxed; see
Nuida et al. \cite{Nuida2010} and Lami et al. \cite{lami2021framework},
for example.} For example, the state of a spin-$1/2$ particle is determined by
the probabilities of the $+1$ outcome of measuring spin observables
in three orthogonal directions, as demonstrated by the Bloch vector
description.

There are many different yet equivalent ways to formulate the GPT
framework. To make this paper self-contained, let us briefly outline
an intuitive approach to GPTs which is based on an operational derivation
\cite{Masanes2011}.

\subsection{States\label{subsec:States}}

If a system has a \emph{minimal} fiducial set consisting of $d$ outcomes\footnote{A fiducial set is minimal if there is no such set with fewer than
$d$ outcomes.}, its \emph{state space} $\mathcal{S}$ is given by a convex, compact
set of vectors of the form 
\begin{equation}
\boldsymbol{\omega}=\begin{pmatrix}p_{1}\\
\vdots\\
p_{d}\\
1
\end{pmatrix}\in\R^{d+1},\label{eq:probstates}
\end{equation}
where $p_{k}\in\left[0,1\right],k=1\ldots d$, are the probabilities
of the fiducial outcomes. The extra dimension of the ``ambient''
vector space simplifies the description of measurement outcomes, as
explained below. The convexity of the state space follows from the
assumption that if one were to prepare the system in the states $\boldsymbol{\omega}$
and $\boldsymbol{\omega}'=\left(p_{1}^{\prime},\ldots,p_{d}^{\prime},1\right)^{T}$
with probabilities $\lambda$ and $(1-\lambda)$, respectively, then
the probability of observing the $k$-th fiducial measurement outcome
should equal 
\begin{equation}
p_{k}^{\prime\prime}(\lambda)=\lambda p_{k}+\left(1-\lambda\right)p_{k}^{\prime}\,,\qquad\lambda\in[0,1],\quad k=1\ldots d\,;\label{eq: mixing of probs}
\end{equation}
therefore, this mixed state should be represented by the vector 
\begin{equation}
\boldsymbol{\boldsymbol{\omega}}''(\lambda)=\lambda\boldsymbol{\omega}+\left(1-\lambda\right)\boldsymbol{\omega}'\,.\label{eq: mixed state}
\end{equation}

A state $\boldsymbol{\omega}$ is \emph{extremal }if it cannot be
written as a (non-trivial) convex combination of other states. The
state space is assumed to be compact since, firstly, it must be bounded
if the entries of the vector are to be between zero and one. Secondly,
as an arbitrarily good approximation of a state would be operationally
indistinguishable from the state itself, we also assume the state
space is closed in the Euclidean topology. Throughout the manuscript
we will use the Euclidean topology and Euclidean norm $\left\Vert \cdot\right\Vert $.

As an example, consider a \emph{classical bit} which may reside in
one of two states called ``0'' and ``1'', or in a mixture of the
two. If we know that the bit is in state 0 with probability $p$ then
it is in state 1 with probability $(1-p)$; in other words, the number
$p\in[0,1]$ determines the state of the system. When performing the
measurement which asks ``Is the bit in state 0 or 1?'', the outcome
``The bit is in state 0.'' forms a complete set of fiducial measurement
outcomes. Thus, the state space $\mathcal{S}_{b}$ of the bit can
be represented by the line segment between $(0,1)^{T}$ and $(1,1)^{T}$,
as displayed in Fig. \ref{fig:skewbit} (see Section. \ref{subsec:Effects-and-observables}).
The end points of the segment correspond to the states 0 and 1, respectively,
and their convex hull defines the state space ${\cal S}_{b}$.

\subsection{Effects and observables\label{subsec:Effects-and-observables}}

The possible outcomes of measuring an observable in a GPT system with
state space $\cS$ correspond to \emph{effects} which are linear maps
$e:\R^{d+1}\rightarrow\R$ such that $0\leq e\left(\boldsymbol{\omega}\right)\leq1$
for all states $\boldsymbol{\omega}\in\mathcal{S}$; here $e\left(\boldsymbol{\omega}\right)$
denotes the probability of observing the outcome $e$ when a measurement
$\mathbb{M}$ (with $e$ as a possible outcome) is performed on a
system in state $\boldsymbol{\omega}$. Due to the linearity of the
map $e$, any effect can be uniquely expressed in the form 
\begin{equation}
e\left(\boldsymbol{\omega}\right)=\boldsymbol{e}\cdot\boldsymbol{\omega},\label{eq: linear map as scpr}
\end{equation}
for some vector $\boldsymbol{e}\in\R^{d+1}$. We will also use the
term ``effect'' to refer to the vector $\boldsymbol{e}$ representing
a map $e$.

The linearity of effects is motivated by the assumption that they
should respect the mixing of states with some parameter $\lambda\in[0,1]$.
More specifically, the following two events should occur with the
same probability: 
\begin{enumerate}
\item observing the outcome $e$ of a measurement $\mathbb{M}$ performed
on a system in a mixed state $\boldsymbol{\omega}''(\lambda)=\lambda\boldsymbol{\omega}+\left(1-\lambda\right)\boldsymbol{\omega}'$; 
\item observing the outcome $e$ when the measurement $\mathbb{M}$ is performed
with probability $\lambda$ on a system in state $\boldsymbol{\omega}$
and with probability $(1-\lambda)$ on a system prepared in state
$\boldsymbol{\omega}'$. 
\end{enumerate}
This assumption implies that the map $e$ should satisfy 
\begin{equation}
e\left(\lambda\boldsymbol{\omega}+\left(1-\lambda\right)\boldsymbol{\omega}'\right)=\lambda e\left(\boldsymbol{\omega}\right)+\left(1-\lambda\right)e\left(\boldsymbol{\omega}'\right)\,,\qquad\boldsymbol{\omega},\boldsymbol{\omega}^{\prime}\in\mathcal{S}\,.\label{eq: affine condition}
\end{equation}
Thus, the map $\E$ is an \emph{affine} function on the state space
$\mathcal{S}$ which can be extended to a \emph{linear} function on
the vector space $\R^{d+1}$ containing $\mathcal{S}$.

The set of all effects associated with measurement outcomes in a specific
GPT system is known as its \emph{effect space,} $\mathcal{E}$. The
space $\mathcal{E}$ corresponds to a convex subset of $\R^{d+1}$,
as does the state space $\mathcal{S}$. It necessarily contains the
zero and unit vectors, 
\begin{equation}
\boldsymbol{0}=\begin{pmatrix}0\\
\vdots\\
0\\
0
\end{pmatrix}\qquad\text{ and }\qquad\boldsymbol{u}=\begin{pmatrix}0\\
\vdots\\
0\\
1
\end{pmatrix},\label{eq: zero and unit}
\end{equation}
as well as the vector $(\boldsymbol{u}-\boldsymbol{e})$ for every
$\boldsymbol{e}\in\mathcal{E}$ \cite{janottanorestriction}, which
arises automatically as a valid effect. We also assume that the effect
space spans the full $(d+1)$ dimensions of the vector space; otherwise
the model would contain states which result in identical probabilities
for all effects in the effect space, making them indistinguishable
and hence operationally equivalent. Note that a $d$-dimensional state
space comes with a $\left(d+1\right)$-dimensional effect space. \emph{Extremal
}effects are defined by the property that they cannot be written as
a (non-trivial) convex combination of other effects.

\emph{Observables} are given by tuples $\left\llbracket \E_{1},\E_{2},\ldots\right\rrbracket $
of elements of the effect space that sum to the unit effect $\boldsymbol{u}$,
with each effect in the tuple corresponding to a different possible
outcome when measuring the observable. The position of an effect in
the tuple encodes the label of the corresponding outcome. Given the
observable $\mathbb{D}_{\E}=\left\llbracket \E,\boldsymbol{u}-\E\right\rrbracket $,
for example, we will say effect $\E$ represents the first possible
outcome of measuring $\mathbb{D}_{\E}$ since $\E$ occupies the first
position in the tuple. A GPT should also specify which tuples of effects
correspond to observables (or, in the language of \cite{Filippov2019},
the GPT should specify the set of \emph{meters}). We will assume throughout
(except in Section. \ref{sec:Simulability}) that any finite tuple
of effects $\left\llbracket \E_{1},\ldots,\E_{n}\right\rrbracket $
satisfying 
\begin{equation}
\sum_{j=1}^{n}\E_{j}=\bu\label{eq:unit}
\end{equation}
and
\begin{equation}
\sum_{j\in J}\E_{j}\in\mathcal{E}\label{eq:coarse}
\end{equation}
for any subset $J\subset\left\{ 1,\ldots,n\right\} $, corresponds
to an observable of the GPT system\footnote{This assumption is equivalent to considering GPTs with a restriction
of type (R1) in \cite{Filippov2019}, whereas the GPTs considered
in Section. \ref{sec:Simulability} can be of type (R2) or (R3).}. Eq. (\ref{eq:coarse}) ensures that the set of observables is closed
under coarse-graining of outcomes.

The effect space $\mathcal{E}_{b}$ of the classical bit with state
space $\mathcal{S}_{b}$ is given by the parallelogram depicted in
Fig. \ref{fig:skewbit}. The two-outcome measurement $\mathbb{B}$
answering ``Is the bit in state 0 or 1?'' is represented by 
\begin{equation}
\mathbb{B}=\left\llbracket \begin{pmatrix}-1\\
1
\end{pmatrix},\begin{pmatrix}1\\
0
\end{pmatrix}\right\rrbracket .\label{eq: two-outcome meas}
\end{equation}

\begin{figure}
\begin{centering}
\subfloat[\label{fig:skewbit}]{\begin{centering}
\includegraphics[scale=2]{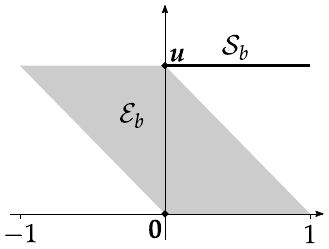}
\par\end{centering}
}\hfill{}\subfloat[\label{fig:bit}]{\begin{centering}
\includegraphics[scale=2]{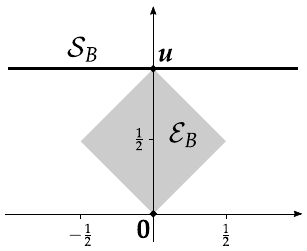}
\par\end{centering}
}
\par\end{centering}
\caption{\label{fig:bits}State and effect spaces of the classical-bit GPT
system: (a) formulated as described in Sections \ref{subsec:States}
and \ref{subsec:Effects-and-observables} and using Eq. (\ref{eq: two-outcome meas});
(b) after applying the transformation given in Eq. (\ref{eq:bittrans}).}
\end{figure}

\subsection{Equivalent GPTs\label{subsec:Transforming-GPTs}}

When considering a specific GPT it is sometimes useful to linearly
transform its state and effect spaces. The Bloch vector describing
a qubit state is a case in point since its components do not necessarily
take values in the range $\left[0,1\right]$. The Bloch vector representation
of a qubit density operator 
\begin{equation}
\rho=\frac{1}{2}\left(\I+x\sigma_{x}+y\sigma_{y}+z\sigma_{z}\right),\label{eq: qubit density matrix}
\end{equation}
is given by the vector $\left(x,y,z,1\right)^{T}$, with the fourth
component being the coefficient of the identity matrix. The linear
relation between each of the coefficients $r\in\left\{ x,y,z\right\} $
and the probability $p_{r}$ of finding the outcome ``$+1$'' when
measuring $\sigma_{r}$ reads explicitly $p_{r}=\left(1+r\right)/2$.

Any linear transformation which preserves the inner product between
states and effects of a given GPT system gives rise to an alternative
representation. Suppose that we transform the state space $\mathcal{S}$
by an invertible $\left(d+1\right)\times\left(d+1\right)$ matrix
$\mathbf{M}$ to the space $\mathcal{S}_{\mathbf{M}}\equiv\mathbf{M}\mathcal{S}$.
Then we must apply the inverse transpose transformation $\mathbf{M}^{-T}\equiv\left(\mathbf{M}^{-1}\right)^{T}$
to the effect space, $\mathcal{E}_{\mathbf{M}}\equiv\mathbf{M}^{-T}\mathcal{E}$,
in order that the probabilities remain invariant, 
\begin{equation}
\E_{\mathbf{M}}\cdot\boldsymbol{\omega}_{\mathbf{M}}=\left(\mathbf{M}^{-T}\E\right)\cdot\left(\mathbf{M}\boldsymbol{\omega}\right)=\E\cdot\boldsymbol{\omega}\,.\label{eq: invariance of probs}
\end{equation}
The transformed state and effect spaces continue to be convex subsets
of $\R^{d+1}$, and they can even be thought of as a convex subset
of a vector space isomorphic to $\R^{d+1}$. GPTs are often presented
in this way (cf. \cite{Barnum2014,janottanorestriction} and references
therein).

The standard formulation of quantum theory in finite dimensions is
an example of representing the state and effect spaces of a theory
as subsets of a vector space isomorphic to $\R^{d+1}$. Quantum states
are represented by density operators on $\C^{d}$ which form a convex
subset of the real vector space of Hermitian operators on $\C^{d}$,
which is isomorphic to $\R^{d^{2}}$. Quantum effects, or elements
of a positive-operator measure (POM), can also be embedded in this
space with $e\left(\boldsymbol{\omega}\right)=\Tr\left(\boldsymbol{e}\boldsymbol{\omega}\right)$
for an operator $\boldsymbol{e}$ satisfying $0\leq\bk{\psi}{\boldsymbol{e}\ket{\psi}}\leq\bk{\psi}{\psi}$
for all rays $\ket{\psi}\in\C^{d}$. Using this representation of
the state and effect spaces is convenient for $d>2$ since it is cumbersome
to explicitly describe the set of density operators by some set of
constraints on vectors of the form given in Eq. (\ref{eq:probstates})
(see e.g. \cite{kimura2003bloch,Goyal_2016,bengtsson2017geometry}). 

As an explicit example, let us transform the GPT description of a
classical bit with state space $\mathcal{S}_{b}$ by the matrix 
\begin{equation}
\mathbf{M}=\begin{pmatrix}2 & -1\\
0 & 1
\end{pmatrix}\,.\label{eq:bittrans}
\end{equation}
The new state space, $\mathcal{S}_{B}\equiv\mathbf{M}\mathcal{S}_{b}$,
is now the convex hull of the images of the extremal states $0$ and
$1$ (previously located at $(0,1)^{T}$ and $(1,1)^{T}$, respectively),
i.e. 
\begin{equation}
\mathcal{S}_{B}=\conv\left\{ \begin{pmatrix}-1\\
1
\end{pmatrix},\begin{pmatrix}1\\
1
\end{pmatrix}\right\} \,.\label{eq: S_B as hull}
\end{equation}
Similarly, the effect space, $\mathcal{E}_{B}\equiv\mathbf{M}^{-T}{\cal E}_{b}$,
is given by the convex hull of the zero effect $\mathbf{0}$, the
unit effect $\bu$ and two other extremal effects, 
\begin{equation}
{\cal E}_{B}=\conv\left\{ \begin{pmatrix}0\\
0
\end{pmatrix},\begin{pmatrix}1\\
0
\end{pmatrix},\frac{1}{2}\begin{pmatrix}-1\\
1
\end{pmatrix},\frac{1}{2}\begin{pmatrix}1\\
1
\end{pmatrix}\right\} \,,\label{eq: E_B as hull}
\end{equation}
as pictured in Figure \ref{fig:bit}.

\subsection{Cones in GPTs \label{subsec:Cones-in-GPTs}}

The notion of a \emph{positive} \emph{cone }is useful when studying
the properties of state and effect spaces of a GPT. A positive cone
is a subset of $\R^{d+1}$ that contains all non-negative linear combinations
of its elements (see \cite{Rockafellar1970}, in which our positive
cone corresponds to a cone containing the origin). Positive cones
may, for example, be generated from convex subsets of real vector
spaces. 
\begin{defn}
\label{def:The-convex-cone}The \textit{positive cone }$A^{+}$ of
a convex subset $A$ of a real vector space is the set of vectors
\begin{equation}
A^{+}=\left\{ x\boldsymbol{a}|x\geq0,\boldsymbol{a}\in A\right\} \,.\label{eq: define positive cone}
\end{equation}
\end{defn}
Positive cones also arise from considering the space dual to a subset
of vectors in an inner product space. 
\begin{defn}
The \emph{dual cone} $A^{*}$ of a subset $A$ of a real inner product
space $V$ is the positive cone 
\begin{equation}
A^{*}=\left\{ \boldsymbol{b}\in V|\left\langle \boldsymbol{a},\boldsymbol{b}\right\rangle \geq0\text{ for all }\boldsymbol{a}\in A\right\} \,.\label{eq: define dual cone}
\end{equation}
\end{defn}
Figure (\ref{fig:bitcone}) illustrates, for a classical bit, the
dual cone $\mathcal{S}_{B}^{*}$ of the state space $\mathcal{S}_{B}$.
It is easy to see that, in general, the effect space ${\cal E}$ of
a GPT system must be contained within the dual cone $\mathcal{S}^{*}$
of the state space in order that the effects assign non-negative probabilities
to every state in the state space.

The following lemma describes a simple but important property of effect
spaces related to the fact that the elements of its dual cone effectively
span the ambient space. 
\begin{lem}
For any effect space $\mathcal{E}$ and any vector $\boldsymbol{c}\in\R^{d+1}$,
we have $\boldsymbol{c}=\boldsymbol{a}-\boldsymbol{b}$ for some vectors
$\boldsymbol{a},\boldsymbol{b}\in\mathcal{E}^{+}$.\label{Decomposition-1} 
\end{lem}
\begin{proof}
Firstly, the interior of $\mathcal{E}^{+}$ is non-empty since $\mathcal{E}$
is convex and spans $\R^{d+1}$. Let $\boldsymbol{e}$ be an interior
point of $\mathcal{E}^{+}$. As $\boldsymbol{e}$ is an interior point
of $\mathcal{E}^{+}$, we have that $\boldsymbol{e}+\epsilon\boldsymbol{c}\in\mathcal{E}^{+}$
for some $\epsilon>0$ and we may take $\boldsymbol{a}=\left(\boldsymbol{e}+\epsilon\boldsymbol{c}\right)/\epsilon$
and $\boldsymbol{b}=\boldsymbol{e}/\epsilon$. 
\end{proof}
\begin{figure}
\subfloat[\label{fig:bitcone}]{\centering{}\includegraphics[scale=1.5]{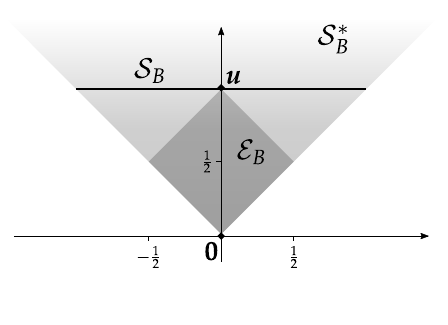}

}\hfill{}\subfloat[\label{fig:bitcones}]{\centering{}\includegraphics[scale=1.5]{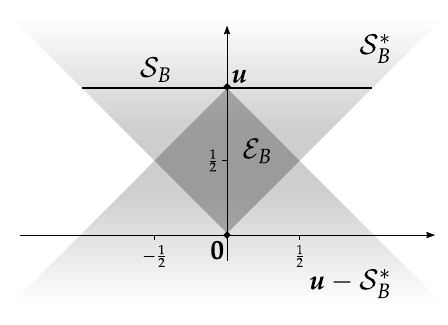}

}\caption{State and effect spaces $\mathcal{S}_{B}$ and $\mathcal{E}_{B}$
(the horizontal line and the dark square, respectively) of the classical-bit
GPT system showing (a) the dual cone $\mathcal{S}_{B}^{*}$ (shaded)
of the state space; (b) the effect space $\mathcal{E}_{B}$ is given
by the intersection of the cones $\mathcal{S}_{B}^{*}$ and $\boldsymbol{u}-\mathcal{S}_{B}^{*}$
since the bit is \emph{unrestricted} (see Section \ref{subsec:no-restriction-hypothesis}).}
\end{figure}

Two further lemmata (proven in Appendix \ref{conelemmata}), which
we will need later on, establish relations between positive cones
and their dual cones. We will denote the closure of a set $A$ by
$\overline{A}$.

\begin{restatable}{lem}{dualcone}

\label{Lemma: dual-cone}Let $A$ be a compact, convex subset of $\R^{d+1}$,
then $A^{**}=\overline{A^{+}}$.

\end{restatable}

\begin{restatable}{lem}{dualconedualofcone}

\label{dualconedualofcone}For a compact and convex subset $A\subset\R^{d+1}$,
we have $A^{*}=\left(A^{+}\right)^{*}$. 

\end{restatable}

Finally, it will also be important that the positive cone of a GPT
state space is always a closed set (also proven in Appendix \ref{conelemmata}).

\begin{restatable}{lem}{Splusclosed}

\label{lem:Splusclosed}Given a GPT state space $\cS$, the positive
cone $\cS^{+}$is closed.

\end{restatable}

\section{Unrestricted GPTs}

The main result of this paper is to establish a Gleason-type theorem
for a class of GPTs which we will introduce in this section, namely\emph{
almost noisy unrestricted }GPTs. First, we describe the no-restriction
hypothesis and the unrestricted GPTs it defines. Next, we define noisy
unrestricted (NU) GPTs and, building on this concept, we define \emph{almost}
NU GPTs.

\subsection{The no-restriction hypothesis \label{subsec:no-restriction-hypothesis}}

A particularly close relationship between state and effect spaces
exists in GPTs that satisfy the \emph{no-restriction hypothesis }\cite{janottanorestriction},\emph{
}i.e. GPTs with effect spaces consisting of \emph{all} linear maps
$e:\R^{d+1}\rightarrow\R$ such that $0\leq e\left(\boldsymbol{\omega}\right)\leq1$
for all $\boldsymbol{\omega}\in\mathcal{S}$. In such an \emph{unrestricted}
theory the state space defines a unique effect space, and \emph{vice
versa}. The effect space of a system with state space ${\cal S}$
in an unrestricted GPT is given by 
\begin{align}
E\left(\mathcal{S}\right) & =\left\{ \boldsymbol{e}\in\R^{d+1}|0\leq\boldsymbol{e}\cdot\boldsymbol{\omega}\leq1,\text{ for all }\boldsymbol{\omega}\in\mathcal{S}\right\} \nonumber \\
 & =\mathcal{S}^{*}\cap\left(\boldsymbol{u}-\mathcal{S}^{*}\right)\,,\label{eq: E map}
\end{align}
where $\boldsymbol{u}-\mathcal{S}^{*}=\left\{ \boldsymbol{u}-\boldsymbol{e}|\boldsymbol{e}\in\mathcal{S}^{*}\right\} $.
The classical bit is an example of an unrestricted GPT system. The
cones $\mathcal{S}^{*}$ and $\left(\boldsymbol{u}-\mathcal{S}^{*}\right)$
as well as their intersection are illustrated in Figure \ref{fig:bitcones}.

Conversely, if an unrestricted GPT system has an effect space $\mathcal{E}$
then a unique state space is associated with it, namely: 
\begin{align}
W\left(\mathcal{E}\right) & =\left\{ \boldsymbol{\omega}\in\R^{d+1}|\boldsymbol{e}\cdot\boldsymbol{\omega}\geq0\text{ for all }\boldsymbol{e}\in\mathcal{E}\text{ and }\boldsymbol{\omega}\cdot\boldsymbol{u}=1\right\} \nonumber \\
 & =\mathcal{E}^{*}\cap\boldsymbol{1}_{d+1}\,,\label{eq: W map}
\end{align}
where $\boldsymbol{1}_{d+1}=\left\{ \boldsymbol{\omega}\in\R^{d+1}|\boldsymbol{u}\cdot\boldsymbol{\omega}=1\right\} $;
we will omit the subscript $d+1$ whenever the dimension is clear
from the context. We have introduced the maps $E$ and $W$ in the
context of unrestricted GPTs but they are well-defined for the state
and effect spaces of any GPT. The maps will play an important role
in the derivation of our main result (see Section \ref{sec:Main-Result}).

\subsection{Noisy unrestricted GPTs\label{subsec:Noisy-unrestricted-GPTs}}

The class of NU GPTs consists of all unrestricted GPTs along with
a special subset of restricted GPTs. The included restricted GPTs
are those that can be thought of as unrestricted GPTs in which some
(or all) of the observables can only be measured with a limited efficiency,
or with some inherent noise.
\begin{defn}
\label{A-NU-GPT}A GPT system is \emph{noisy unrestricted} whenever
its state space $\cS$ and effect space $\mathcal{E}$ satisfy the
following property: for every vector $\boldsymbol{e}\in E\left(\mathcal{S}\right)$
there exists a number $p_{\E}\in\left(0,1\right]$ such that the rescaled
vector $p_{\E}\boldsymbol{e}$ is contained in the effect space $\mathcal{E}$.
\end{defn}
It follows that NU GPT systems are exactly those in which the positive
cones of the effect space $\mathcal{E}$ and the unrestricted effect
space $E(\cS)$ are equal. Moreover, Definition \ref{A-NU-GPT} is
equivalent to the statement that each NU GPT system is closely related
to an unrestricted GPT system in the following way: for each observable
$\mathbb{O}=\left\llbracket \E_{1},\E_{2},\ldots,\E_{n}\right\rrbracket $
in the unrestricted GPT system, there exists some $p\in\left(0,1\right]$
and an observable 
\begin{equation}
\mathbb{O}_{p}=\left\llbracket p\E_{1},p\E_{2},\ldots,p\E_{n},\left(1-p\right)\boldsymbol{u}\right\rrbracket ,\label{eq:Op}
\end{equation}
in the NU GPT system, while the state spaces of the two systems are
given by the same set. Thus, measuring the observable $\mathbb{O}_{p}$
of the NU GPT can be thought of as successfully measuring the observable
$\mathbb{O}$ (of the associated unrestricted GPT) with probability
$p$, and observing no outcome with probability $(1-p)$, regardless
of the state of the system. In the language of \cite{postselection},
a measurement of $\mathbb{O}_{p}$ can \emph{simulate} a measurement
of $\mathbb{O}$ when one allows for post-selection on the outcomes.
The case of $p_{\E}=1$ for all vectors in $\E\in E\left(\mathcal{S}\right)$
is included in Definition \ref{A-NU-GPT} for later convenience; in
other words, ``noiseless'' unrestricted GPTs---i.e. those in which
$\mathcal{E}=E\left(S\right)$ holds---are also considered to be
NU GPTs. All other NU GPTs, however, are restricted, i.e. they violate
the no-restriction hypothesis.

Figure \ref{fig:NU and R bits} shows two modified versions of the
bit GPT system that violate the no-restriction hypothesis, one of
which is a NU GPT system while the other is not. Further examples
of the three different varieties of GPT system---restricted, unrestricted
and noisy unrestricted---can be found in Section \ref{sec:Examples}.
\begin{figure*}
\centering{}\subfloat[\label{fig:NUbit}]{\centering{}\includegraphics[scale=2]{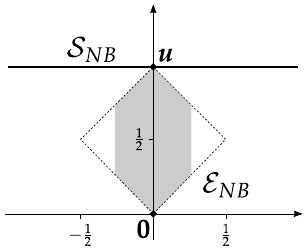}

}\hfill{}\subfloat[\label{fig:Rbit}]{\centering{}\includegraphics[scale=2]{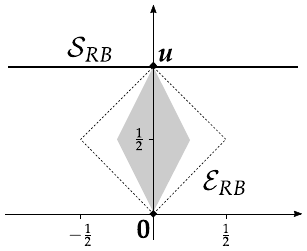}

}\caption{\label{fig:NU and R bits}State and effect spaces for two restrictions
of the classical-bit GPT system: (a) resulting in a NU GPT system
and (b) \emph{not} resulting in a NU GPT system (see Definition \ref{A-NU-GPT}).}
\end{figure*}

\subsection{Almost noisy unrestricted GPTs}

We now introduce almost NU GPTs as a relaxation of the class of NU
GPTs. Just as a NU GPT, a given almost NU GPT will be related to a
unique unrestricted GPT. However, given an observable $\mathbb{O}$
of the unrestricted GPT, it will be sufficient for the almost NU GPT
to include a noisy measurement of an observable \emph{arbitrarily
close to} $\mathbb{O}$. In NU GPTs, there is a noisy measurement
for each such observable, $\mathbb{O}$, so the requirement is met.
In a non-NU GPT it can only be met if the GPT has extremal effects
arbitrarily close to the zero effect. For example, consider the almost
NU bit depicted in Fig. \ref{fig:almostNUbit}. Every point on the
boundary of the effect space is extremal and as the non-zero extremal
effects approach zero they also approach the boundary of $\mathcal{E}_{B}^{+}$,
the positive cone of the unrestricted bit effect space.

\begin{figure}
\begin{centering}
\includegraphics[scale=2]{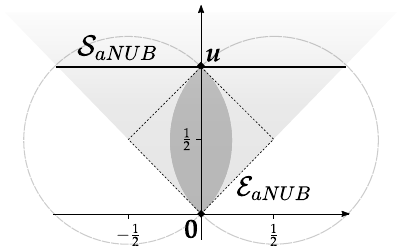}
\par\end{centering}
\caption{\label{fig:almostNUbit}State and effect spaces of a restriction of
the classical bit GPT that is almost NU but not NU. The effect space
$\mathcal{E}_{aNUB}$ (dark grey) is given by the intersection of
two discs (boundaries shown as dashed lines) centred at the two extremal
(but not $\bO$ or $\bu$) points of the unrestricted bit effect space
$\mathcal{\mathcal{E}}_{B}$ (dashed square). The positive cone of
$\mathcal{E}_{B}$ (light grey) is equal to the closure of the positive
cone of $\mathcal{E}_{aNUB}$.}
\end{figure}

In general an almost NU GPT is defined as follows.
\begin{defn}
\label{def:AlmostNU}A GPT system is \emph{almost} noisy unrestricted
(aNU) whenever its state space ${\cal S}$ and effect space $\mathcal{E}$
satisfy the following condition: for every vector $\boldsymbol{e}\in E\left(\mathcal{S}\right)$
and every $\epsilon>0$ there exists a vector $\boldsymbol{e}_{\epsilon}\in E(\mathcal{S})$
and a number $p_{\E}^{\epsilon}\in\left(0,1\right]$ such that $\left\Vert \boldsymbol{e}-\boldsymbol{e}_{\epsilon}\right\Vert <\epsilon$
and the rescaled vector $p_{\E}^{\epsilon}\boldsymbol{e}_{\epsilon}$
is contained in the effect space $\mathcal{E}$.
\end{defn}
Alternatively, we can characterise aNU GPTs as those that can simulate
any observable in their unrestricted counterpart arbitrarily well
via post-selection. Explicity, for every observable $\mathbb{O}=\left\llbracket \E_{1},\E_{2},\ldots,\E_{n}\right\rrbracket $
in the unrestricted GPT system and $\epsilon>0$, there exists a second
observable (the approximation) $\mathbb{O}_{\epsilon}=\left\llbracket \E_{1}',\E_{2}',\ldots,\E_{n}'\right\rrbracket $
of the unrestricted GPT system satisfying $\left\Vert \E_{j}-\E_{j}'\right\Vert \leq\epsilon$
for all $1\leq j\leq n$, and also a probability $p>0$ such that
$\mathbb{O}'=\left\llbracket p\E_{1}',p\E_{2}',\ldots,p\E_{n}',(1-p)\bu\right\rrbracket $
is an observable of the aNU GPT system. The fact that Def. \ref{def:AlmostNU}
is a necessary condition for a GPT to have this property is clear.
To show that it is also sufficient we need the following alternative
characterisation of aNU GPTs which will also be key to proving our
main result.

\begin{restatable}{lem}{aNU}

\label{lem:An-almost-NU}In an almost NU GPT the state space ${\cal S}$
and effect space ${\cal E}$ of each system are related by $E\left(\mathcal{S}\right)=\overline{\mathcal{E}^{+}}\cap\left(\boldsymbol{u}-\overline{\mathcal{E}^{+}}\right)$.

\end{restatable}

The proof can be found in Appendix \ref{conelemmata}. Note that,
as shown in the proof of Lemma \ref{lem:An-almost-NU}, in an aNU
GPT we find $\overline{\mathcal{E}^{+}}=E(\cS)^{+}$, a condition
considered by Ludwig in his operational framework, see \cite[Chapt.  VI, Thm. 2.2.1]{ludwighilbertspace}.

Now we can show that an aNU GPT system contains an observable $\mathbb{O}'$
which is a noisy version of an arbitrarily close approximation $\mathbb{O}_{\epsilon}$
to any observable $\mathbb{O}$ in its unrestricted counterpart. Let
$\cS$ and $\mathcal{E}$ be state and effect spaces of a system in
an almost NU GPT and $E(\cS)$ be the effect space of the corresponding
unrestricted GPT system. Firstly, if $\mathbb{O}=\left\llbracket \bu/n,\bu/n,\ldots,\bu/n\right\rrbracket $
we may take $\mathbb{O}'=\mathbb{O}$. Secondly, for every $\mathbb{O}=\left\llbracket \E_{1},\E_{2},\ldots,\E_{n}\right\rrbracket \neq\left\llbracket \bu/n,\bu/n,\ldots,\bu/n\right\rrbracket $
in the unrestricted system and $\epsilon>0$ define $\E_{j}'=\delta\bu/n+(1-\delta)\E_{j}$
where 
\[
\delta=\min\left\{ \left.\frac{\epsilon}{\left\Vert \E_{j}-\bu/n\right\Vert }\right|1\leq j\leq n\text{ such that }\E_{j}\neq\bu/n\right\} 
\]
 for all $1\leq j\leq n$. Since $E(\cS)$ spans $\R^{d+1}$ and contains
$\bu-\E$ for all $\E\in E(\cS)$ it follows that $\bu/n$ and hence
$\E'_{j}$ are interior points of $E(\cS)^{+}$. By the equality $\overline{\mathcal{E}^{+}}=E(\cS)^{+}$
we have that $\E'_{j}\in\mathcal{E}^{+}$ and thus the observable
\[
\mathbb{O}'=\left\llbracket p\E_{1}',p\E_{2}',\ldots,p\E_{n}',(1-p)\bu\right\rrbracket 
\]
 is in the aNU GPT for some $p>0$. This observable is a noisy version
of the observable $\mathbb{O}_{\epsilon}=\left\llbracket \E_{1}',\E_{2}',\ldots,\E_{n}'\right\rrbracket $
which satisfies $\left\Vert \E_{j}-\E_{j}'\right\Vert \leq\epsilon$
for all $1\leq j\leq n$.

Note that, compared to the condition in Lemma \ref{lem:An-almost-NU},
NU GPTs satisfy the stronger condition $E\left(\mathcal{S}\right)=\mathcal{E}^{+}\cap\left(\boldsymbol{u}-\mathcal{E}^{+}\right)$.
Hence in GPTs which are aNU but \emph{not} NU we find $\mathcal{E}^{+}\neq\overline{\mathcal{E}^{+}}$,
i.e. the positive cones of the effect spaces are not closed. It follows
from the proof of Lemma \ref{lem:Splusclosed} that $\mathcal{E}^{+}\neq\overline{\mathcal{E}^{+}}$
only holds when there are extremal effects arbitrarily close to the
zero effect, such as in the aNU bit in Fig. \ref{fig:almostNUbit}.

\section{Gleason-type theorems for GPTs \label{sec:Main-Result}}

Gleason's theorem is motivated by the idea the probabilities of the
outcomes of any measurement performed on a quantum system should uniquely
define the state of the system. Thus, every state should come with
a \emph{frame function,} that is a probability assignment on the space
of projections (and later, quantum effects \cite{Busch2003}) such
that the probabilities of the disjoint outcomes of any measurement
sum to unity. In order to formulate a GTT for GPTs, we need to generalize
the concept of a frame function.
\begin{defn}
\label{A-generalized-probability}A \emph{frame function} on the effect
space $\mathcal{E}$ of a GPT system is a map $v:\R^{d+1}\rightarrow\R$
satisfying
\begin{enumerate}
\item[(V1)] $0\leq v\left(\boldsymbol{e}\right)\leq1$ for all effects $\E\in\mathcal{E}$;\label{enu:v1} 
\item[(V2)] $v\left(\E_{1}\right)+v\left(\E_{2}\right)+\ldots+v\left(\E_{n}\right)=1$
for all sequences of effects $\E_{1},\E_{2},\ldots,\E_{n}\in\mathcal{E}$
such that $\left\llbracket \E_{1},\E_{2},\ldots,\E_{n}\right\rrbracket $
is an observable in the GPT.
\end{enumerate}
\end{defn}
Considering measurements with only a finite number of possible outcomes
is sufficient for our purposes; thus, assumption (V2) is only required
to hold for finite sequences of effects. Countable sequences of effects
may be required if one considers infinite-dimensional systems.

In quantum theory the results of Gleason and Busch show that any frame
function must correspond to a density operator. In other words, there
are no states beyond those we already believe to exist under the assumption
that states must correspond to frame functions. We will take the analog
of this idea as the definition of a GTT for a GPT. 
\begin{defn}
\label{def:-GPT admitting a GTT} A GPT \emph{admits a Gleason-type
theorem} if and only if for each of its systems every frame function
on the effect space $\mathcal{E}$ can be represented by a state in
the state space $\mathcal{S}$.
\end{defn}
In this definition, the state space $\cS$ is prescribed by the GPT
and may be a subset of the set $W(\mathcal{E})$ of all mathematically
well-defined states given the effect space $\mathcal{E}$. The existence
of a GTT would allow the set of all possible states of a GPT system
to follow from the effect space via the natural assumption that a
state can be uniquely defined by its propensity to take each possible
value of every observable. The requirement that all mathematically
possible states are realised in a theory could be thought of as dual
to the no-restriction hypothesis, i.e. requiring that all effects
have a corresponding measurement outcome. We will show, however, that
the classes of GPTs that satisfy these requirements do \emph{not}
coincide.

\subsection{A Gleason-type theorem for almost NU GPTs\label{subsec:A-Gleason-type-theorem}}

After these preliminaries, let us state the main result of this paper
which identifies the condition under which Gleason-type theorems exist
for general probabilistic theories. 
\begin{thm}
\label{thm: NU main}Let $\mathcal{S}$ and $\mathcal{E}$ be the
state and effect spaces, respectively, of a GPT system. Any frame
function $v:\mathcal{E}\rightarrow\left[0,1\right]$ admits an expression
$v\left(\boldsymbol{e}\right)=\boldsymbol{e}\cdot\boldsymbol{\omega}$
for some $\boldsymbol{\omega}\in\mathcal{S}$ if and only if

\begin{equation}
\overline{\mathcal{E}^{+}}\cap\left(\boldsymbol{u}-\overline{\mathcal{E}^{+}}\right)=E\left(\mathcal{S}\right),\label{eq:NU}
\end{equation}
i.e. a GPT admits a Gleason-type theorem if and only if it is an almost
noisy unrestricted GPT.
\end{thm}
Since quantum theory in finite dimensions is a GPT obeying the no-restriction
hypothesis, Busch's result \cite{Busch2003} is an immediate consequence
of Theorem \ref{thm: NU main}. The infinite-dimensional case, however,
will not be treated here.

\subsection{Consequences of a GTT for aNU GPTs}

Before presenting the proof of Theorem \ref{thm: NU main}, we briefly
explain how a GTT allows one to simplify the postulates used to describe
a specific GPT in an axiomatic approach. A simple way to state the
postulates, often used for quantum theory, is to describe the mathematical
objects that represent \emph{observables} and \emph{states} along
with the rule for calculating the \emph{probabilities} of measurement
outcomes (supplemented by postulates describing the composition of
systems and, possibly, the evolution of the system in time). In general,
for some GPT system with effect space $\mathcal{E}$ and state space
$\mathcal{S}$, such postulates would take the following form:
\begin{enumerate}
\item[(O)] The \emph{observables} of the system correspond exactly to the tuples
of vectors $\left\llbracket \boldsymbol{e}_{1},\boldsymbol{e}_{2},\ldots\right\rrbracket $
in $\mathcal{E}$ that sum to the vector $\boldsymbol{u}$, with each
vector corresponding to a possible disjoint outcome of measuring the
observable. 
\item[(S)] The \emph{states} of the system correspond exactly to vectors $\boldsymbol{\omega}\in\mathcal{S}$. 
\item[(P)] When measuring the observable $\left\llbracket \boldsymbol{e}_{1},\boldsymbol{e}_{2},\ldots\right\rrbracket $
on a system in state $\boldsymbol{\omega}\in\mathcal{S}$, the \emph{probability}
to obtain outcome $\boldsymbol{e}_{j}$ is given by $p_{j}(\boldsymbol{\omega})=\boldsymbol{e}_{j}\cdot\boldsymbol{\omega}$. 
\end{enumerate}
If there exists a GTT for the GPT in hand, then it could be recovered
by replacing the postulates (S) and (P) by the \emph{operationally
motivated} assumption that every state must have a corresponding frame
function defining its outcome probabilities, along with the converse
assumption that every frame function must have a corresponding state
in the theory\emph{. }Consequently, one only needs to supplement postulate
(O) with a single new postulate. 
\begin{enumerate}
\item[(F)] There exists a state of the system for every frame function on the
effect space $\mathcal{E}$. 
\end{enumerate}
Combined with the GTT and operational reasoning, postulates (O) and
(F) lead to the same theory as the postulates (O), (S) and (P). Thus,
the Gleason-type theorem would simplify the axiomatic formulation
of the GPT, just as the theorems by Gleason or Busch do in the case
of quantum theory.

It could be argued that the operational assumptions present in the
derivation of the GPT framework are better suited to simplifying the
postulates (O), (S) and (P). In Appendix \ref{sec: alternative simplification of axions for a GPT}
we compare this strategy with the simplification achieved using Theorem
\ref{thm: NU main}.

Our result also opens up an alternative step in establishing the GPT
framework. In Section \ref{sec:Generalized-probabilistic-theori}
we reviewed the derivation of the framework from operational principles
which motivates the structure of \emph{state} \emph{spaces} in GPTs
from the assumption of fiducial measurement outcomes, as in \cite{Hardy2001a,BarrettGPT,Masanes2011,short2010strong,Barnum2014,Sainz2018}.
The approach then arrives at the effect-space structure by motivating
effects as affine functions on the states space. 

One may, however, invert this process and arrive at the same framework.
By assuming that fiducial \emph{states} exist one can motivate the
structure of \emph{effect} \emph{spaces} in GPTs\footnote{The approach of motivating the structure of measurements first features
in other formalisms; for example, see Ludwig's work on operational
theories, summarised in \cite{ludwighilbertspace}, or the theory
of test-spaces and its predecessors \cite{foulis1981empirical,randall1973operational,foulis1993logicoalgebraic}.}, namely as convex, compact subsets of a real vector space containing
the zero vector and a vector $\bu$ such that $\bu-\E$ is in the
set for every effect $\E$\footnote{This structure predates the GPT framework, see \cite[Chapt.~IV, Thm.~1.2]{ludwighilbertspace}.}.
At this point, the standard approach would motivate states as affine
functionals on the effect space using arguments based on classical
mixtures.

Our results offer a mathematically weaker alternative: it is sufficient
to combine the effect-space structure with a minimal set of two-outcome
observables and their convex combinations as described in Section
\ref{sec:Simulability}. Then, a corollary to our Gleason-type theorem
may be used to recover the structure of the state space (as a subset
of $W\left(\mathcal{E}\right)$ where $\mathcal{E}$ is the effect
space of the model) from the assumption that a state must come with
a unique frame function on this minimal set of observables.

In other words, assuming that states are frame functions offers an
alternative to the mathematically stron\-ger assumption that states
are linear functionals on the real vector space containing the effect
space. The details of this approach are described in Appendix~\ref{sec:Deriving-the-GPT}.

We highlight one observation following from this derivation, namely,
the inequivalence of the no-restriction hypothesis---as described
in Section. \ref{subsec:no-restriction-hypothesis}---and the dual
assumption in the fiducial-states approach, which we will call the
\emph{no-state-re\-stric\-tion hypothesis}\footnote{Following this terminology the no-restriction hypothesis should be
more accurately called the no-\emph{effect}-restriction hypothesis.
For consistency with the literature and brevity we will stick with
with the standard nomenclature.}. Recall that the no-restriction hypothesis assumes that, given a
state space $\cS$, every mathematically valid effect is indeed an
effect of the system, i.e. $\mathcal{E}=E(\cS)$. The no-state-re\-stric\-tion
hypothesis says that, given an effect space $\mathcal{E}$, every
mathematically valid state appears in the theory, i.e. $\cS=W(\mathcal{E})$.
As shown in the proof of Theorem \ref{thm: NU main} below, the relation
$\cS=W(\mathcal{E})$ holds for exactly almost NU GPTs. Therefore
we can conclude that the no-restriction and no-state-restriction hypotheses
are \emph{not} equivalent and the former constitutes a stronger assumption.

Since both approaches are operationally valid, the no-restriction
hypothesis and the no-state-restriction hypothesis appear to be motivated
equally well. One could argue that their inequivalence demonstrates
that neither assumption is valid. In any case, if one is ready to
assume either of them, one should be willing to assume the other,
too. Thus, the stronger no-effect-restriction hypothesis should be
taken, as seems to have happened naturally in the literature, e.g
in \cite{BarrettGPT,short2010strong,janotta2011limits,lamithesis,lami2021framework}.
Conversely, they should also both be rejected together, meaning no-go
theorems such as in \cite{Sainz2018} would benefit from targetting
the weaker no-state-restriction hypothesis.

\subsection{Proof of Theorem \ref{thm: NU main}}

Using Definition \ref{def:-GPT admitting a GTT} of a GTT, Theorem
\ref{thm: NU main} shows that aNU GPTs are exactly the class of GPTs
that admit GTTs. We will prove this result in two steps: (i) in Lemma
\ref{thm:Main}, a frame function on a GPT effect space $\mathcal{E}$
is found to correspond to a vector in the set $W\left(\mathcal{E}\right)$
defined in Section \ref{subsec:no-restriction-hypothesis}; (ii) the
set $W\left(\mathcal{E}\right)$ is found to correspond to the state
space of a GPT system if and only if the GPT is in the class of aNU
GPTs, in Lemmata \ref{Lemma WES} and \ref{Lemma: EWE}.

Step (i): The proof of the following lemma\footnote{This lemma was independently shown in \cite{farid2019}, where it
is considered as a GTT for GPTs.} follows the one for the quantum case given in \cite{Busch2003}.

\begin{restatable}{lem}{prop}

\label{thm:Main}Let $\mathcal{E}$ be an effect space of a GPT. Any
frame function $v$ on $\mathcal{E}$ admits an expression 
\begin{equation}
v\left(\E\right)=\E\cdot\boldsymbol{\omega},
\end{equation}
for some vector $\boldsymbol{\omega}\in W\left(\mathcal{E}\right)$
and all effects $\E\in\mathcal{E}$. 

\end{restatable}

Lemma \ref{thm:Main} shows that if one defines states as frame functions
on an effect space $\mathcal{E}$, then the associated state space
must be $W\left(\mathcal{E}\right)$.

Step (ii): We will now prove that the set $W\left(\mathcal{E}\right)$
corresponds to the state space of a GPT system with effect space $\mathcal{E}$
if and only if the GPT is an aNU GPT. Two lemmata will be needed to
show that $W\left(E\left(\mathcal{S}\right)\right)=\mathcal{S}$ holds
for \emph{all} GPTs while the relation $E\left(W\left(\mathcal{E}\right)\right)=E\left(\mathcal{S}\right)$
only holds for aNU GPTs. The proofs can be found in Appendix \ref{WESandEWE}.

\begin{restatable}{lem}{WES}

\label{Lemma WES}For any GPT system with state space $\mathcal{S}$,
we have $W\left(E\left(\mathcal{S}\right)\right)=\mathcal{S}$.

\end{restatable}

\begin{restatable}{lem}{EWE}

\label{Lemma: EWE}Given a GPT system with state and effect spaces
$\mathcal{S}$ and $\mathcal{E}$, respectively, the relation $E\left(W\left(\mathcal{E}\right)\right)=E\left(\mathcal{S}\right)$
holds if and only if $E\left(\mathcal{S}\right)=\overline{\mathcal{E}^{+}}\cap\left(\boldsymbol{u}-\overline{\mathcal{E}^{+}}\right)$.

\end{restatable}

We are now in a position to prove our main result, Theorem \ref{thm: NU main},
announced in the previous section. It states that a general probabilistic
theory admits a Gleason-type theorem if and only if it is \emph{almost
noisy unrestricted}. The result is an immediate consequence of the
lemmata just shown.
\begin{proof}
By Lemma \ref{thm:Main} we know that frame functions must be of the
form $v\left(\boldsymbol{e}\right)=\boldsymbol{e}\cdot\boldsymbol{\omega}$
for some $\boldsymbol{\omega}\in W\left(\mathcal{E}\right)$. Thus,
to conclude the proof we need to show that the state space of a GPT
system coincides with the set $W\left(\mathcal{E}\right)$ exactly
in almost NU GPTs, i.e. $W\left(\mathcal{E}\right)=\mathcal{S}$,
if and only if Eq. (\ref{eq:NU}) holds.

Let $W\left(\mathcal{E}\right)=\mathcal{S}'$. Firstly, assume that
Eq. (\ref{eq:NU}) holds. By Lemma \ref{Lemma: EWE} we have 
\begin{equation}
E\left(\mathcal{S}'\right)=E\left(W\text{\ensuremath{\left(\mathcal{E}\right)}}\right)=E\left(\mathcal{S}\right).
\end{equation}
Now, by applying the $W$ map to both sides of this equation and using
Lemma \ref{Lemma WES}, we find 
\begin{equation}
W\left(E\left(\mathcal{S}'\right)\right)=\mathcal{S}'=W\left(E\left(\mathcal{S}\right)\right)=\mathcal{S}.
\end{equation}
Secondly, assume that Eq. (\ref{eq:NU}) does not hold, i.e. $\overline{\mathcal{E}^{+}}\cap\left(\boldsymbol{u}-\overline{\mathcal{E}^{+}}\right)\neq E\left(\mathcal{S}\right)$,
then by Lemma \ref{Lemma: EWE} 
\begin{equation}
E\left(\mathcal{S}'\right)\neq E\left(\mathcal{S}\right),
\end{equation}
and $\mathcal{S}'\neq\mathcal{S}$, which is the content of Theorem
\ref{thm: NU main}. 
\end{proof}

\section{Examples and applications\label{sec:Examples}}

In this section, we will consider examples of NU GPT systems to show
how their axiomatic formulation simplifies due to the Gleason-type
theorem they allow. We also highlight that a simple well-known non-quantum
model introduced by Spekkens \cite{spekkens2007evidence} does not
belong to the class of almost NU GPTs and, therefore, does not come
with a GTT.

\subsection{Simplified axioms for a rebit and other unrestricted GPTs\label{subsec:Unrestricted-GPTs}}

Unrestricted GPTs are a well-studied class of GPT. The \emph{rebit
}\cite{caves2001entanglement}, for example, is a GPT system with
a disc-shaped state space. The state space can be equivalently modeled
(see Section \ref{subsec:Transforming-GPTs}) by the subset of\emph{
real }density matrices of a qubit. Rebits are convenient low-dimensional
building blocks for a toy model of quantum theory, giving rise to
many characteristic features such as superposition, entanglement and
non-locality \cite{wootters2012entanglement,caves2001entanglement,janotta2011limits}.

Using our notation, the state space of a rebit is given by

\begin{equation}
\mathcal{S}_{R}=\conv\left\{ \boldsymbol{\omega}_{\theta}\right\} _{\theta\in[0,2\pi)},\label{eq: rebit state space}
\end{equation}
where 
\begin{equation}
\boldsymbol{\omega}_{\theta}=\begin{pmatrix}\cos\theta\\
\sin\theta\\
1
\end{pmatrix}.
\end{equation}
The convex hull of the zero effect $\boldsymbol{0}$, the unit effect
$\boldsymbol{u}$ and a continuous ring of effects 
\begin{equation}
\mathbf{e}_{\theta}=\frac{1}{2}\begin{pmatrix}\cos\theta\\
\sin\theta\\
1
\end{pmatrix},\quad\theta\in[0,2\pi)\,,\label{eq: ring of extremal effects}
\end{equation}
form the rebit effect space

\begin{equation}
\mathcal{E}_{R}=\conv\left\{ \boldsymbol{0},\boldsymbol{u},\mathbf{e}_{\theta}\right\} _{\theta\in[0,2\pi)},\label{eq: rebit effect space}
\end{equation}
illustrated in Figure \ref{fig:Rebit-state-and}. The rebit satisfies
the no-restriction hypothesis since the effect space $\mathcal{E}_{R}$
is as large as is permitted in the GPT framework (see Eq. \ref{eq: E map}).

\begin{figure}
\centering{}\subfloat[\label{fig:Rebit-state-and}]{\begin{centering}
\includegraphics[scale=1.5]{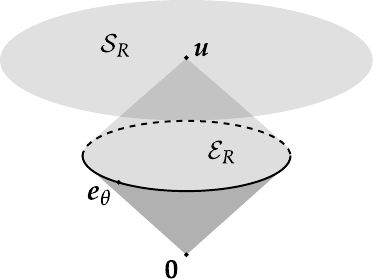}
\par\end{centering}
}\subfloat[\label{fig:squit-1}]{\begin{centering}
\includegraphics[scale=1.5]{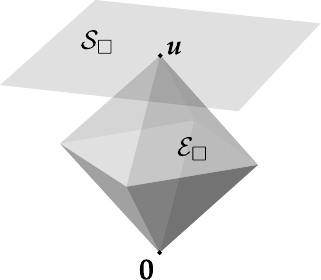}
\par\end{centering}
}\caption{State and effect spaces: (a) of the rebit GPT and (b) of the squit
GPT.}
\end{figure}

Let us now apply the general argument given at the end of Section
\ref{subsec:A-Gleason-type-theorem} to a rebit, as a first example
of an unrestricted GPT. Instead of using the GPT framework, we could
have described the rebit of this hypothetical world in an axiomatic
fashion, i.e. by assuming the axioms (O), (S) and (P) from Section
\ref{subsec:A-Gleason-type-theorem}, using the effect and state spaces
$\mathcal{E}_{R}$ and $\mathcal{S}_{R}$. Then, Theorem \ref{thm: NU main}
states that, alternatively, we could postulate the rebit observables
and, by considering the frame functions associated with them, recover
both the state space $\mathcal{S}_{R}$ and the probability rule.
More explicitly, we replace postulates (S) and (P) by a single postulate
with operational motivation. 
\begin{enumerate}
\item[(F)] The states of a rebit correspond exactly to the frame functions on
the effect space $\mathcal{E}_{R}$. 
\end{enumerate}
In other words, we effectively introduce the states of the rebit as
probability assignments on the outcomes of measurements. The model
created by the postulates (O) and (F) is equivalent to the the original
one in the sense that it makes exactly the same predictions.

Classical bits, qubits and qudits as well as\emph{ square bits }(or\emph{
squits}, for short) are other unrestricted GPTs for which identical
arguments also result in a smaller set of axioms by means of our Gleason-type
theorem. The state and effect spaces of bits and qudits were described
in Section \ref{subsec:Transforming-GPTs} while a \emph{squit} or
\emph{gbit} \cite{BarrettGPT} is a GPT with a square state space
and an octahedral effect space, as illustrated in Figure \ref{fig:squit-1}.
Pairs of squits are often considered in the study of non-local correlations
since they are capable of producing the super-quantum correlations
of a PR-box\emph{ }\cite{Popescu1994}.

\subsection{A GPT with a GTT: the noisy rebit \label{subsec: simplified noisy rebit axioms}}

Next, let us consider a \emph{noisy }rebit characterized by the property
that the extremal rebit observables $\mathbb{D}_{\mathbf{e}_{\theta}}=\left\llbracket \boldsymbol{e}_{\theta},\boldsymbol{u}-\boldsymbol{e}_{\theta}\right\rrbracket ,\theta\in[0,2\pi),$
can be measured only \emph{imperfectly}, i.e. with some efficiency
$p\in\left(0,1\right)$; see Eq. (\ref{eq: ring of extremal effects})
for the definition of the effects $\mathbf{e}_{\theta}$\footnote{The noisy rebit GPT also results from applying the general (shifted)
depolarizing channel for a generic state to the rebit effect space;
see \cite{Filippov2019} for details of the analogous qubit case.}. The state space of this NU GPT system coincides with that of the
rebit, $\mathcal{S}_{R}^{n}\equiv\mathcal{S}_{R}$. In order to define
its effect space , let us introduce two continuous rings of effects,
\begin{equation}
\mathbf{e}_{\theta}^{+}=\frac{p}{2}\begin{pmatrix}\cos\theta\\
\sin\theta\\
1
\end{pmatrix}\quad\text{and}\quad\mathbf{e}_{\theta}^{-}=\frac{p}{2}\begin{pmatrix}\cos\theta\\
\sin\theta\\
2/p-1
\end{pmatrix},\quad\theta\in[0,2\pi)\,.
\end{equation}
These rings, along with the zero effect $\boldsymbol{0}$ and the
unit effect $\bu$, form the extremal points of the noisy rebit effect
space, 
\begin{equation}
\mathcal{E}_{R}^{n}=\conv\left\{ \boldsymbol{0},\boldsymbol{u},\mathbf{e}_{\theta}^{+},\mathbf{e}_{\theta}^{-}\right\} _{\theta\in[0,2\pi)},\label{eq: noisy rebit effect space}
\end{equation}
depicted in Figure \ref{fig:Noisy-rebit-effect}. While still being
a GPT system, the model does not satisfy the no-restriction hypothesis:
the effect space $\mathcal{E}_{R}^{n}$ is \emph{restricted }to a
proper subset of $\mathcal{E}_{R}$ shown in Figure \ref{fig:Rebit-state-and}.
Nevertheless, Theorem \ref{thm: NU main} continues to apply: the
noisy rebit admits a GTT which is effectively due to the fact that
there exist finite neighbourhoods of the zero effect and the unit
effect in which $\mathcal{E}_{R}^{n}$ and $\mathcal{E}_{R}$ coincide.
Thus, the noisy rebit does not satisfy the no-restriction hypothesis,
but it does satisfy the dual assumption of the no-\emph{state}-restriction
hypothesis (see Appendix \ref{sec:Deriving-the-GPT}).

Repeating the argument presented in Section \ref{subsec:Unrestricted-GPTs},
we are able to simplify the definition of the noisy rebit in terms
of postulates (O), (S), and (P) which introduce its effect space $\mathcal{E}_{R}^{n}$,
its state space $\mathcal{S}_{R}^{n}$, and the Born rule, respectively.
The alternative axiomatic formulation in terms of only two postulates
only rests on the effect space of the system, 
\begin{enumerate}
\item[(O)] The \emph{observables} of a noisy rebit correspond exactly to the
tuples of vectors $\left\llbracket \boldsymbol{e}_{1},\boldsymbol{e}_{2},\ldots\right\rrbracket $
in $\mathcal{E}_{R}^{n}$ that sum to the vector $\boldsymbol{u}$,
with each vector corresponding to a possible disjoint outcome of measuring
the observable. 
\item[(F)] The \emph{states} of a noisy rebit correspond exactly to frame functions
on the effect space $\mathcal{E}_{R}^{n}$. 
\end{enumerate}
\emph{Mutatis mutandis}, this procedure applies to any other aNU GPT.

\begin{figure}
\centering{}\includegraphics[scale=1.5]{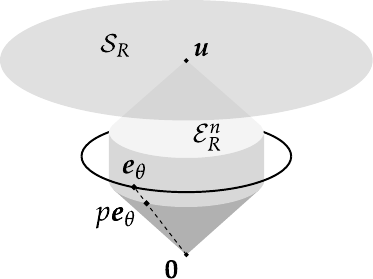}\caption{\label{fig:Noisy-rebit-effect}The effect space ${\cal E}_{R}^{n}$
and the state space $\mathcal{S}_{R}$ of the noisy rebit, an example
of a low-dimensional NU GPT; for comparison, the ring of extremal
rebit effects $\E_{\theta}$ is shown as a thick line.}
\end{figure}

\subsection{A GPT without a GTT: the Spekkens toy model}

In 2007, a toy theory was introduced \cite{spekkens2007evidence}
capable of reproducing a number of important quantum features such
as the existence of non-commuting observables, the impossibility of
cloning arbitrary states and the presence of entanglement while simultaneously
admitting a description in terms of local hidden variables. Originally,
Spekkens' model had been introduced without reference to the GPT framework.
Here, we will consider the ``convexification'' of this model such that it becomes a GPT system, as described
in \cite{janottanorestriction}.

Considered as a GPT system, Spekkens' model comes with a restricted
effect space, and it is not part of an aNU GPT which can be seen as
follows. Its state space is given by a regular octahedron 
\begin{equation}
\mathcal{S}_{S}=\conv\left\{ \boldsymbol{x}_{\pm},\boldsymbol{y}_{\pm},\boldsymbol{z}_{\pm}\right\} ,\label{eq: Spekkens state space}
\end{equation}
with vertices 
\begin{equation}
\boldsymbol{x}_{\pm}=\begin{pmatrix}\pm1\\
0\\
0\\
1
\end{pmatrix},\quad\boldsymbol{y}_{\pm}=\begin{pmatrix}0\\
\pm1\\
0\\
1
\end{pmatrix},\quad\boldsymbol{z}_{\pm}=\begin{pmatrix}0\\
0\\
\pm1\\
1
\end{pmatrix}.\label{eq:spekkens states}
\end{equation}
Under the no-restriction hypothesis the extremal effects (other than
the zero and unit effects) associated with the space $\mathcal{S}_{S}$
would be the vertices of a cube. In Spekkens' model, however, they
are taken to be the vertices of another octahedron inscribed into
this cube, as depicted in Figure \ref{fig:Spekkens}. More explicitly,
the effect space is the convex hull of the zero and unit effects and
the six extremal effects given by the (rescaled) vectors in Eq. (\ref{eq:spekkens states}),
\begin{equation}
\mathcal{E}_{S}=\conv\left\{ \boldsymbol{0},\boldsymbol{u},\frac{\boldsymbol{x}_{\pm}}{2},\frac{\boldsymbol{y}_{\pm}}{2},\frac{\boldsymbol{z}_{\pm}}{2}\right\} .\label{eq: Spekkens effect space}
\end{equation}

Not being an aNU GPT system, Theorem \ref{thm: NU main} tells us
that the toy model does not admit a GTT.\footnote{This fact was, of course, clear without Theorem \ref{thm: NU main}
since every vector in $W\left(\mathcal{E}_{S}\right)\supset\mathcal{S}_{S}$
yields frame function by definition.} It is impossible to reproduce this GPT system by assuming that the
states of the system are in one-to-one correspondence with the frame
functions on the effect space. There are, in fact, \emph{more} frame
functions than states in $\mathcal{S}_{S}$. The frame functions correspond
to all vectors in the set $W\left(\mathcal{E}_{S}\right)$ which is
a strict superset of $\mathcal{S}_{S}$ forming a cube around $\mathcal{S_{S}}$,
in the same way that $E\left(\mathcal{S}_{S}\right)$ encloses $\mathcal{E}_{S}$
in Figure \ref{fig:Spekkens}.

In order to recover the original model, one would have to place a
restriction on which frame functions correspond to allowed states.
This restriction can be considered analogous to relaxing the no-restriction
hypothesis on the effect space.

\begin{figure}
\begin{centering}
\includegraphics[scale=2]{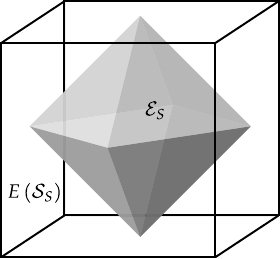}
\par\end{centering}
\caption{The octahedral effect space $\mathcal{E}_{S}$ of Spekkens' toy theory
projected into the hyperplane of $\R^{4}$ obtained by fixing the
fourth entry of the vectors to $1/2$; $\mathcal{E}_{S}$ is a proper
subset of the cubic effect space $E\left(\mathcal{S}_{S}\right)$
(boundary depicted by black lines) required by the no-restriction
hypothesis given the state space $\mathcal{S}_{S}$ defined in Eq.
(\ref{eq: Spekkens state space}).\label{fig:Spekkens}}
\end{figure}

\section{A GTT for almost NU GPTs based on two-outcome observables \label{sec:Simulability}}

The definition of a frame function used in Section \ref{sec:Main-Result}
is based on the idea that \emph{every} sequence of effects $\E_{1},\E_{2},\ldots\in\mathcal{E}$
satisfying Eqs. (\ref{eq:unit}) and (\ref{eq:coarse}) corresponds
to an observable. This assumption is, however, not operationally motivated
and is therefore not part of the most general GPT framework. In this
section we show that relaxing this assumption does not pose an obstacle
to the existence of Gleason-type theorems. In quantum theory a Gleason-type
theorem can already be derived by involving only a specific subset
of \emph{all} POMs \cite{Wright2018} known as projective-simulable
observables \cite{Oszmaniec2017}. We will show now that a similar
weakening of the assumptions continues to imply the result of Theorem
\ref{thm: NU main} in the context of GPTs. In this way, we are able to
extend our result to the most general type of GPTs.

Let us begin by introducing the idea of \emph{simulating} the measurement
of an observable by measuring other observables. This simulation is
achieved in a GPT by classically mixing observables and post-processing
measurement outcomes \cite{Heinosaarisimulable}. For example, to
simulate the observable 
\begin{equation}
\mathbb{G}=\left\llbracket \frac{1}{3}\left(\E_{1}+\E_{2}+2\boldsymbol{f}\right),\boldsymbol{u}-\frac{1}{3}\left(\E_{1}+\E_{2}+2\boldsymbol{f}\right)\right\rrbracket .
\end{equation}
we may measure the observables 
\begin{equation}
\mathbb{E}=\left\llbracket \E_{1},\E_{2},\boldsymbol{u}-\E_{1}-\E_{2}\right\rrbracket \quad\text{and}\quad\mathbb{F}=\left\llbracket \boldsymbol{f},\boldsymbol{0},\boldsymbol{u}-\boldsymbol{f}\right\rrbracket 
\end{equation}
with probabilities $1/3$ and and $2/3$, respectively, to simulate
\begin{equation}
\mathbb{G}^{\prime}=\frac{1}{3}\mathbb{E}+\frac{2}{3}\mathbb{F}=\left\llbracket \frac{1}{3}\left(\E_{1}+2\boldsymbol{f}\right),\frac{1}{3}\E_{2},\boldsymbol{u}-\frac{1}{3}\left(\E_{1}+\E_{2}+2\boldsymbol{f}\right)\right\rrbracket \,
\end{equation}
followed by coarse-graining the first two outcomes to produce the
dichotomic observable $\mathbb{G}$, where the addition and scalar
multiplication of observables is performed elementwise. The only post-processing
necessary in the proof to follow is to add outcomes to an observable
that occur with probability zero. For example, the two-outcome observable
$\left\llbracket \E,\boldsymbol{u}-\E\right\rrbracket $ simulates
the three-outcome observable $\left\llbracket \E,\boldsymbol{u}-\E,\boldsymbol{0}\right\rrbracket $
if one considers there to be a third outcome of measuring $\left\llbracket \E,\boldsymbol{u}-\E\right\rrbracket $
which never occurs.

Now consider the set of observables which may be simulated by two-outcome
extremal observables, i.e. those described by an extremal effect $\E$
and its complement $\boldsymbol{u}-\E$. For brevity we will refer
to such observables as \emph{simulable}. We now define simulable observables,
denoting by $\mathbb{O}(j)$ the $j$-th outcome of an observable
$\mathbb{O}$. A dichotomic observable has precisely two \emph{non-zero
}effects and, possibly, any number of copies of the zero effect.
\begin{defn}
A simulable $n$-outcome observable $\mathbb{O}$ satisfies 
\[
\mathbb{O}(k)=\sum_{j=1}^{N}q(k|j)\mathbb{M}(j)
\]
 for some probability distribution $q(k|j)$ such that $\sum_{k=1}^{N}q(k|j)=1$
for $1\le j\le m$ and some classical mixture $\mathbb{M}$ of dichotomic
extremal observables.
\end{defn}
The set of simulable observables is minimal in the sense that it is
contained in any set of observables, under the operational assumption
that a set of observables be simulation-closed \cite{Filippov2019}.
To see this we first note that each extremal effect must be included
in some observable, otherwise it would be excluded from the effect
space. Any observable containing a given extremal effect $\E$ can
be coarse-grained to give the two-outcome extremal observable $\left\llbracket \E,\boldsymbol{u}-\E\right\rrbracket $.
Thus, in order to be closed under simulation a set of observables
must at least contain all those that can be simulated by two-outcome
extremal observables.

Next, let us call\emph{ }a frame function \emph{simulable }if the
property (V2) in Definition \ref{A-generalized-probability} is required
to hold for simulable observables only.
\begin{defn}
\label{simulable frame function}A \emph{simulable frame function}
on an effect space $\mathcal{E}$ is a map $v:\R^{d+1}\rightarrow\R$
satisfying 
\begin{enumerate}
\item[(S1)] $0\leq v\left(\boldsymbol{e}\right)\leq1$ for all effects $\E\in\mathcal{E}$; 
\item[(S2)] $v\left(\E_{1}\right)+v\left(\E_{2}\right)+\ldots+v\left(\E_{n}\right)=1$
for all sequences of effects $\E_{1},\E_{2},\ldots,\E_{n}\in\mathcal{E}$
which give rise to simulable observables $\mathbb{O}=\left\llbracket \E_{1},\E_{2},\ldots,\E_{n}\right\rrbracket $. 
\end{enumerate}
Theorem \ref{thm: NU main} can now be strengthened because the properties
of simulable frame functions are sufficient for a proof. 
\end{defn}
\begin{thm}
\label{cor:simGPT}Let $\mathcal{S}$ and $\mathcal{E}$ be the state
and effect spaces, respectively, of a NU GPT. Any simulable frame
function $v$ on $\mathcal{E}$ admits an expression 
\begin{equation}
v\left(\E\right)=\E\boldsymbol{\cdot\omega},\label{eq:gleason sim}
\end{equation}
for some $\boldsymbol{\omega}\in\mathcal{S}$ and all $\E\in\mathcal{E}$. 
\end{thm}
\begin{proof}
See Appendix~\ref{sec:Proof-of-Corollary}. 
\end{proof}
This theorem can be used to provide an alternative step in the operational
derivation of the GPT framework as described in Section \ref{subsec:A-Gleason-type-theorem},
with full details given in Appendix~\ref{sec:Deriving-the-GPT}.

\section{Summary and Discussion \label{sec:Summary-and-Discussion}}

From a conceptual point of view, the results of this paper imply that
each general probabilistic theory belongs to one of two distinct classes:
either it admits, like quantum theory, a Gleason-type theorem which
allows us to construct the set of the possible states of the theory,
or it does not admit a GTT.

In Lemma \ref{thm:Main} (see Section \ref{sec:Main-Result}) frame
functions were found to be linear functionals on the effect space.
If one considers this fact to be the main content of the Gleason-type
theorems in quantum theory then the lemma proves that Gleason-type
theorems exist for \emph{all} GPTs. In this paper we have, however,
taken the view that a Gleason-type theorem should establish a bijection
between frame functions and states in the theory under consideration.

Interpreting GTTs in this way, Theorem \ref{thm: NU main} shows that
a GPT admits such a theorem if and only if it is an \emph{almost noisy
unrestricted} GPT, of which classical and quantum models are examples.
Requiring that there is a state in a theory for every frame function
could be considered as dual to the no-restriction hypothesis which
demands that to every mathematical effect there should correspond
a measurement outcome. However, we have shown that the no-restriction
hypothesis is more restrictive than requiring the existence of a GTT. On the one hand, every unrestricted GPT admits a GTT but, on the other, there are almost NU GPTs that admit a GTT but violate the no-restriction
hypothesis. Phrased differently, the no-effect-restriction hypothesis
is manifestly different from the no-state-restriction hypothesis.

In Section \ref{subsec:A-Gleason-type-theorem} we describe how a
Gleason-type theorem can be used to derive the state space in a given
GPT from the set of observables. The postulates (O), (S) and (P),
which specify a given GPT, can be replaced by just two postulates,
namely (O) and (F), when the description of states as frame functions
is assumed. This reduction is only possible in almost NU GPTs.

Extensions of Gleason's theorem to beyond quantum theory have been
considered previously in the literature. Gudder et al. \cite{gudder1999convex} consider
states on convex effect algebras. These algebras can be represented
\cite{gudder1998representation} by subsets $K^{+}\cap(u-K^{+})$
of a real linear space $V$ in which $K$ is a positive cone and $u$
is an element of $K^{+}$. This representation coincides with unrestricted
effect spaces in our terminology. Morphisms (which coincide with frame
functions in our terminology) were shown in \cite{gudder1999convex}
to extend to positive linear functionals on $V$. Barnum \cite{barnum2003quantum}
pointed out that this result can be considered as a Gleason-type theorem
demonstrating, in our terminology, the existence of a Gleason-type
theorem for unrestricted GPTs. Theorem \ref{thm: NU main} extends
this result to exactly the class of almost NU GPTs.

In recent work \cite{Masanes2019,Galley2017,Galley2018}, alternatives
to simplifying the postulates of quantum theory have been put forward
by assuming, for example, the postulates of pure states and their
dynamics in combination with operational reasoning. It would be interesting
to study whether similar approaches also hold for other GPTs, or whether
they are unique to quantum theory.

The current work relies heavily on the convex structure of GPTs. In
future work we would like to establish which GPTs admit an analog
of Gleason's original theorem, in the sense that the frame functions
would only be defined on\emph{ extremal} effects where convexity arguments
can no longer be made.

Finally, it might be possible to establish a link between Gleason-type
theorems and the set of \emph{almost-quantum }correlations \cite{navascues2015almost}.
It is known that GPTs satisfying the no-restriction hypothesis cannot
produce the set of almost quantum\emph{ }correlations in Bell scenarios
\cite{Sainz2018}. If this result could be extended to almost NU GPTs
then the existence of a GTT for a GPT would also preclude the possibility
of that GPT producing the set of almost quantum correlations.
\begin{acknowledgement*}
We are grateful to an anonymous referee for pointing out a mistake
in a draft of our manuscript and for suggesting a way to correct it
which triggered the extension of the result from NU GPTs to aNU GPTs.
VJW acknowledges support by the Foundation for Polish Science (IRAP
project, ICTQT, contract no. 2018/MAB/5, co-financed by EU within
Smart Growth Operational Programme) and the Government of Spain (FIS2020-TRANQI and Severo Ochoa CEX2019-000910-S), Fundaci\'o Cellex, Fundaci\'o Mir-Puig, Generalitat de Catalunya (CERCA, AGAUR SGR 1381 and QuantumCAT).
\end{acknowledgement*}
\bibliographystyle{unsrturl}
\bibliography{gleasonbib}

\onecolumn\newpage
\appendix

\section{Proofs of Lemmata \ref{Lemma: dual-cone}, \ref{dualconedualofcone}
, \ref{lem:Splusclosed} and \ref{lem:An-almost-NU}\label{conelemmata}}

\dualcone*
\begin{proof}
Firstly, let $\ba\in\overline{A^{+}}$. Then there exists a sequence
$(\ba_{n})_{n\in\mathbb{N}}$ such that $\ba_{n}\to\ba$ as $n\to\infty$
with $\ba_{n}\in A^{+}$. For any $\boldsymbol{b}\in A^{*}$ we have
$\ba_{n}\cdot\boldsymbol{b}\geq0$ and, since the inner product is
continuous, $\ba\cdot\boldsymbol{b}\geq0$. Therefore we have shown
$\ba\in A^{**}$.

Secondly, consider $\bx\notin\overline{A^{+}}$. By the hyperplane
separation theorem there exists $\bh\in\R^{d+1}$ and real numbers
$c_{1}>c_{2}$ such that $\ba\cdot\bh\geq c_{1}$ for all $\ba\in\overline{A^{+}}$
and $\bx\cdot\bh\leq c_{2}$. For all $\ba\in\overline{A^{+}}$ and
$\lambda>0$ we have $\lambda\ba\cdot\bh\geq c_{1}$ and therefore
$\ba\cdot\bh\geq c_{1}/\lambda$. Taking the limit as $\lambda\to\text{\ensuremath{\infty}}$
we find $\ba\cdot\bh\geq0$ for all $\ba\in\overline{A^{+}}$ and
hence $\bh\in A^{*}$. Finally, since $\bO\in\overline{A^{+}}$ we
find $0=\bO\cdot\bh\geq c_{1}>c_{2}$, and thus $\bx\cdot\bh\leq c_{2}<0$,
meaning $\bx\notin A^{**}$ since we found that $\bh\in A^{*}$.
\end{proof}
\dualconedualofcone*
\begin{proof}
By Definition \ref{Lemma: dual-cone}, a vector $\boldsymbol{b}$
is in the dual cone $A^{*}$ of $A$ if and only if $\boldsymbol{b}\cdot\boldsymbol{a}\geq0$
for all $\boldsymbol{a}\in A$. Equivalently, we may require $x\left(\boldsymbol{b}\cdot\boldsymbol{a}\right)=\boldsymbol{b}\cdot\left(x\boldsymbol{a}\right)\geq0$
for all vectors $\boldsymbol{a}$ in the set $A$ and $x\geq0,$ which
holds if and only if $\boldsymbol{b}\in A^{+}$. 
\end{proof}
\Splusclosed*
\begin{proof}
Firstly, note that by definition the set $\cS$ does not contain the
zero vector, $\boldsymbol{0}$. Let $(\bx_{j})_{j\in\mathbb{\mathbb{N}}}$
be a convergent sequence in $\R^{d+1}$ such that $\bx_{j}\in\cS^{+}$
for all $j\in\mathbb{N}$ and $\bx_{j}\to\bx$ as $j\to\infty$. We
will show that $\bx\in\cS^{+}$. 

If $\bx=0$, the statement holds. Now assume $\bx\neq0$. For each
$j\in\mathbb{N}$ there exists $p_{j}>0$ and $\bo_{j}\in\cS$ such
that $\bx_{j}=p_{j}\bo_{j}$. By the Bolzano-Weierstrass theorem and
the compactness of $\cS$, the sequence $(\bo_{j})_{j\in\mathbb{N}}$
has a convergent subsequence $(\bo_{k(j)})_{j\in\mathbb{N}}$ such
that $\bo_{k(j)}\to\bo\in\cS$ as $j\to\infty$. 

Now we show that the corresponding subsequence $(p_{k(j)})_{j\in\mathbb{N}}$
neither diverges to infinity nor tends to zero. Firstly, assume that
$p_{k(j)}\to\infty$ as $j\to\infty$. Then, since $\bx_{k(j)}\to\bx$,
we find the contradiction $\bo_{k(j)}=\bx_{k(j)}/p_{k(j)}\to\bO\notin\cS$.
Secondly, assume that $p_{k(j)}\to0$ as $j\to\infty$. Then $\bx_{k(j)}=p_{k(j)}\bo_{k(j)}\to\bO$,
contradicting our assumption that $\bx\neq\bO$. 

Therefore, since $(\bx_{k(j)})_{j\in\mathbb{\mathbb{N}}}$ converges,
we find there exists a finite number $p>0$ such that $p_{k(j)}\to p$
as $j\to\infty$. This result means that in the limit of $j\to\infty$,
$\bx_{j}$ will converge to an element of $\cS^{+}$: $\bx_{j}=p_{j}\bo_{j}\to p\bo\in\cS^{+}$
as $j\to\infty$.
\end{proof}
\aNU*
\begin{proof}
Consider a GPT system with state space $\mathcal{S}$ and effect space
$\mathcal{E}$. Assume that the spaces satisfy $E\left(\mathcal{S}\right)=\overline{\mathcal{E}^{+}}\cap\left(\boldsymbol{u}-\overline{\mathcal{E}^{+}}\right)$. 

Then if $\boldsymbol{e}\in E\left(\mathcal{S}\right)$, we have $\boldsymbol{e}\in\overline{\mathcal{E}^{+}}$
and for every $\epsilon>0$ there exists $\E_{\epsilon}\in\mathcal{E}^{+}$
such that $\left\Vert \E-\E_{\epsilon}\right\Vert <\epsilon$. Hence
there exists $p\in\left(0,1\right]$ such that $p\boldsymbol{e}_{\epsilon}\in\mathcal{E}$
which means that the GPT system satisfies Definition \ref{def:AlmostNU}.

Conversely, assume that the GPT satisfies Definition \ref{def:AlmostNU}.
We will show that $\overline{\mathcal{E}^{+}}=E(\cS)^{+}$, by firstly
showing $E(\cS)^{+}\subseteq\overline{\mathcal{E}^{+}}$ and secondly
that $\overline{\mathcal{E}^{+}}\subseteq E(\cS)^{+}$.

For every vector $b\E\in E(\mathcal{S})^{+}$ (where $\E\in E(\mathcal{S})$
and $b>0$) and $\epsilon>0$ there exists $a\E_{\epsilon}\in\mathcal{E}^{+}$
(where $\E_{\epsilon}\in\mathcal{E}$ and $a>0$) such that $\left\Vert b\E-ba\E_{\epsilon}\right\Vert <b\epsilon$.
Letting $\epsilon=\delta/b$, we find for every $\F\in E(\mathcal{S})^{+}$
and $\delta>0$ there exists $\F_{\epsilon}\in\mathcal{E}^{+}$ such
that $\left\Vert \F-\F_{\epsilon}\right\Vert <\epsilon$. Since every
point in $E(\cS)^{+}$ is a point of closure of $\mathcal{E}^{+},$
we find $E(\cS)^{+}\subseteq\overline{\mathcal{E}^{+}}$.

To prove that $\overline{\mathcal{E}^{+}}\subseteq E(\cS)^{+}$, we
begin by showing that $E(\cS)^{+}$ is a closed set. Firstly, we have
$E(\cS)^{+}=\left(\mathcal{S}^{*}\cap\left(\boldsymbol{u}-\mathcal{S}^{*}\right)\right)^{+}$
by Eq. (\ref{eq: E map}). Then we find 
\begin{equation}
\left(\mathcal{S}^{*}\cap\left(\boldsymbol{u}-\mathcal{S}^{*}\right)\right)^{+}=\cS^{*}\label{eq:ESplusSstar}
\end{equation}
as follows. The set on the left of this equation is clearly contained
in that on the right since we have $\left(\mathcal{S}^{*}\cap\left(\boldsymbol{u}-\mathcal{S}^{*}\right)\right)^{+}\subseteq\left(\mathcal{S}^{*}\right)^{+}=\mathcal{S}^{*}$.
Furthermore, if $\boldsymbol{e}\in\mathcal{S}^{*}$, then non-negative
rescalings of $\boldsymbol{e}$ are also contained in $\mathcal{S}^{*}$:
$x\boldsymbol{e}\in\mathcal{S}^{*}$ for all $x\geq0$. Since $\boldsymbol{u}\cdot\boldsymbol{\omega}=1$
for all $\boldsymbol{\omega}\in\mathcal{S}$, $\boldsymbol{u}$ is
an internal point of $\mathcal{S}^{*}$. Thus, there exists an open
ball $\mathfrak{B}\left(\boldsymbol{u},\epsilon\right)$ around $\boldsymbol{u}$
of radius $\epsilon$ in $\mathcal{S}^{*}$ for some $\epsilon>0$.
Therefore, for $x<\epsilon/\norm{\boldsymbol{e}}$ we have $\norm{\boldsymbol{u}-\left(\boldsymbol{u}-x\boldsymbol{e}\right)}<\epsilon$
and hence $\boldsymbol{u}-x\boldsymbol{e}\in\mathcal{S}^{*}$. By
definition, $x\boldsymbol{e}\in\left(\boldsymbol{u}-\mathcal{S}^{*}\right)$,
hence we have $x\boldsymbol{e}\in\mathcal{S}^{*}\cap\left(\boldsymbol{u}-\mathcal{S}^{*}\right)$
and $\boldsymbol{e}\in\left(\mathcal{S}^{*}\cap\left(\boldsymbol{u}-\mathcal{S}^{*}\right)\right)^{+}$,
thus verifying Eq. (\ref{eq:ESplusSstar}). The dual cone $\cS^{*}$
is, by definition, the intersection of a collection of closed halfspaces
therefore $\cS^{*}=E(\cS)^{+}$ is closed.

Now, since $E(\cS)^{+}$ is closed and contains $\mathcal{E}^{+}$,
we have $\overline{\mathcal{E}^{+}}\subseteq E(\cS)^{+}$. Finally,
we have shown $\overline{\mathcal{E}^{+}}=E(\cS)^{+}$ which gives
$\overline{\mathcal{E}^{+}}\cap\left(\boldsymbol{u}-\overline{\mathcal{E}^{+}}\right)=E\left(\mathcal{S}\right)^{+}\cap(\bu-E\left(\mathcal{S}\right)^{+})=E\left(\mathcal{S}\right)$.
\end{proof}

\section{An alternative simplification of the axioms defining a GPT\label{sec: alternative simplification of axions for a GPT}}

Let us briefly mention an alternative approach to simplifying the
postulates (S), (O) and (P) providing an equivalent definition a GPT
which closely follows the operational assumptions of the GPT framework.
It should be compared with the simplification using Theorem \ref{thm: NU main}
outlined at the end of Section \ref{subsec:A-Gleason-type-theorem}.
The starting point is a single postulate about the states of the model
at hand. 
\begin{enumerate}
\item[(S')] There exist $d$ fiducial measurement outcomes of observables whose
probabilities determine the state of the system. These states are
restricted to being represented by vectors in $\mathcal{S}$. 
\end{enumerate}
The first part of the postulate, the existence of $d$ fiducial measurement
outcomes, determines that the state space can be embedded in $\R^{d}$
and is convex, with convex combinations of vectors representing classical
mixtures of the corresponding states. However, this assumption does
not determine the ``shape'' of the state space, hence the inclusion
of the second part of the postulate restricting the state space to
$\mathcal{S}$. For a specific GPT, the second part of the postulate
may take a more natural-sounding form such as state vectors having
modulus less than or equal to one. From (S'), using the standard operational
assumption that effects must respect classical mixtures and the no-restriction
hypothesis (see Section \ref{sec:Generalized-probabilistic-theori}),
the postulates (O) and (P) are recovered easily.

Let us conclude by comparing this approach to our approach of using
Theorem \ref{thm: NU main} in order to reduce the postulates (O),
(S) and (P).

First, postulate (O) does not assume that there exists $d$ fiducial
outcomes. This property is a consequence in our approach once the
states are identified as linear functionals on the effect space. Therefore,
postulate (O) is not simply a stronger version of (S').

Second, in order to postulate the existence of $d$ fiducial measurement
outcomes, as is done in (S'), one assumes some knowledge of all the
observables of the system; otherwise one would not know that the two
outcomes in question form a complete fiducial set. Therefore, axiom
(S') makes assumptions about both the states and the observables of
the system whereas (O) only concerns observables.

Finally, in the approach based on (S'), additional assumptions would
be necessary to reconstruct an almost NU GPT which does not satisfy
the no-restriction hypothesis because one could not use the no-restriction
hypothesis to recover the postulates (O) and (P). However, such a
GPT does admit a GTT, as Theorem \ref{thm: NU main} shows, and hence
the first method would still be valid.

\section{The ``fiducial state'' derivation of the GPT framework \label{sec:Deriving-the-GPT}}

In the modern literature, the GPT framework is typically derived,
as in Section \ref{sec:Generalized-probabilistic-theori}, by assuming
the existence of fiducial measurement outcomes first, then defining
the state space of a system, followed by a full treatment of observables
and their measure\-ment, see for example \cite{Hardy1999,Masanes2011,BarrettGPT,short2010strong,Barnum2014,Sainz2018}.
However, one may equally consider the inverted argument, i.e. derive
the framework using equivalent operational assumptions, by assuming
a fiducial set of states, in order to define all possible measure\-ments
and their outcomes then finding the compatible mathematical description
of states. Note that such a dual approach is not novel, for example,
Ludwig describes the idea in his work on operational theories \cite{ludwighilbertspace}
exemplify and the test-space formalism \cite{foulis1981empirical,foulis1993logicoalgebraic}
begins with the structure of measurement outcomes first. Proceeding
in this second, dual manner the structure of effect spaces is established
first then Theorem \ref{cor:simGPT} presents an alternative method
for deriving the structure of state spaces, compared with the standard
argument involving mixtures of measure\-ment outcomes.

We begin by summarising the ``fiducial states'' derivation of the
GPT frame\-work in parallel with Section \ref{sec:Generalized-probabilistic-theori}.
Consider all the possible outcomes of the measure\-ments of all the
observables of a given system. We will assume that there exists a
finite set of \emph{fiducial states} such that any one of these outcomes,
$\zeta$, is uniquely determined by the probabilities of $\zeta$
being observed after a measurement (of which $\zeta$ is a possible
outcome) is performed on the system in each of the fiducial states.
In other words, for a system with $d$ states in its fiducial set,
an outcome may be identified by the vector $\E\in\R^{d}$ such that
\begin{equation}
\E=\begin{pmatrix}p_{1}\\
\vdots\\
p_{d}
\end{pmatrix},
\end{equation}
where $p_{j}$ is the probability of observing the outcome for a system
in the $j$th fiducial state. This representation of measurement outcomes
is derived from the operational assumption that one should be able
to distinguish two distinct measurement out\-comes by their statistics
on a finite number of states, in analogy to assuming the possibility
of distinguishing two distinct states from the probabilities of a
finite number of measurement outcomes in the ``fiducial measurements''
approach.

In line with GPT terminology we will call the set of vectors corresponding
to outcomes in a model the effect space and the vectors within this
set effects. Note that the effects are now simply vectors and not
linear functionals. For brevity, we will often refer to a measurement
outcome as the effect by which it is represented.

In the bit example from Section \ref{sec:Generalized-probabilistic-theori},
the fiducial set of states could be the ``0'' and ``1'' states.
Thus the effect space would be a subset of $\R^{2}$.

We will assume the existence of an outcome that occurs with probability
one for any state of the system. This outcome must be represented
by the effect 
\begin{equation}
\boldsymbol{u}=\begin{pmatrix}1\\
\vdots\\
1
\end{pmatrix}.
\end{equation}
Similarly, we assume the existence of an outcome that never occurs,
repres\-ented by the effect 
\begin{equation}
\boldsymbol{0}=\begin{pmatrix}0\\
\vdots\\
0
\end{pmatrix}.
\end{equation}
Any outcome $\E$ must have a complement, namely the outcome ``not
$\E$'' necessarily occurring with probability $(1-p_{j})$ when
the measurement of ``$\E$ or not $\E$'' is performed on the $j$th
fiducial state. Therefore, for any effect $\E=\left(p_{1},\ldots,p_{d}\right)^{T}$
the vector 
\begin{equation}
\boldsymbol{u}-\E=\begin{pmatrix}1-p_{1}\\
\vdots\\
1-p_{d}
\end{pmatrix}
\end{equation}
must also be in the effect space.

Consider two measurements on the system each with a discrete set of
possible outcomes and label the outcomes of each measurement with
positive integers such that the first measurement has outcomes $\left\{ \E_{1},\E_{2},\ldots\right\} $
and the second $\left\{ \E'_{1},\E'_{2},\ldots\right\} $ (if the
measurement has a finite number, $n$, of possible outcomes the labels
$j$ for $j>n$ are assigned the zero effect). If a classical mixture
of these measurements is performed then possible outcomes of this
procedure can be represented by convex combinations of effects. Specifically,
if the first measurement is performed with probability $p$ and the
second with probability $1-p$, then observing an outcome labeled
$j$ from this procedure must be represented by the vector $p\E_{j}+(1-p)\E'_{j}$
in order to be consistent with the fiducial state set. Therefore we
assume the effect space is convex. Finally, since an arbitrarily good
approximation of an effect would ope\-ra\-tio\-nally be indistinguishable
from the effect itself we assume the effect space is a closed subset
of $\R^{d}$.

Returning to the bit example, we can build our effect space from the
requirement of having a measurement that perfectly distinguishes ``0''
and ``1'', and must there\-fore have outcomes, $(1,0)^{T}$ and
$(0,1)^{T}$. Combined with the other requirements for an effect space
we find the bit effect space to be the square in Figure \ref{fig:effbit},
a trans\-formation of the bit effect space described in Section \ref{subsec:Effects-and-observables}.

\begin{figure}
\begin{centering}
\includegraphics[width=0.4\textwidth]{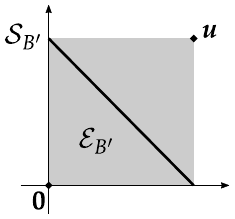}
\par\end{centering}
\caption{\label{fig:effbit}State and effect spaces $\mathcal{S}_{B'}$ (diagonal
black line) and $\mathcal{E}_{B'}$ (grey square), respectively, of
the classical-bit GPT when formulated in the ``measurement-first''
method.}
\end{figure}

We have arrived at the same requirements for the structure of an effect
space as were descri\-bed in Section \ref{sec:Generalized-probabilistic-theori}
(a convex, compact subset of a real vector space containing the zero
vector, and a vector $\bu$ such that $\bu-\E$ is in the set for
every $\E$ in the set). We may now consider how states should be
represented in the frame\-work. We assume a state will be represented
by a map $\omega$ from an outcome $\E$ to the probability of observing
$\E$ when a measurement (of which $\E$ is a possible outcome) is
performed on a system in state $\omega$. From here we may derive
the state space structure of the GPT framework using the standard
operational assumptions or the alternative presented by Theorem \ref{cor:simGPT}.

One the one hand, the standard method for deriving the structure of
the state space is to exploit the fact that we wish for outcome probabilities
to respect mixtures, in analogy with the reasoning behind (\ref{eq: affine condition}),
to find 
\begin{equation}
\omega\left(p\E+(1-p)\E'\right)=p\omega\left(\E\right)+(1-p)\omega\left(\E'\right),
\end{equation}
for $p\in\left[0,1\right]$ and all effects $\E,\E'$. Thus each map
$\omega$ admits an expression 
\begin{equation}
\omega(\E)=\E\cdot\bo,
\end{equation}
for all effects $\E$ and some $\bo\in W\left(\mathcal{E}\right)\in\R^{d}$.

One the other hand, we have already assumed that a pair $\left\{ \E,\boldsymbol{u}-\E\right\} $
form a measurement and have introduced the formalism for describing
mixtures of meas\-ure\-ments, therefore the simulable measurements
from Section \ref{sec:Simulability} are already inclu\-ded in the
framework. Theorem \ref{cor:simGPT} then tells us that if a state
$\omega$ is to assign probab\-ilities to the possible outcomes of
these measurements such that the probabilities of all the outcomes
sum to one then 
\begin{equation}
\omega(\E)=\E\cdot\bo,
\end{equation}
for all effects $\E$ and some $\bo\in W\left(\mathcal{E}\right)\in\R^{d}$.

Both of these approaches lead to the conclusion that the state space
of a GPT with effect space $\mathcal{E}$ must be a subset of $W\left(\mathcal{E}\right)$.
Although the conditions are mathem\-atically different there is no
clear conceptual advantage to either argument.

The ``fiducial states'' derivation of the framework highlights the
existence of a relative of the no-restriction hypothesis, which we
will call the \emph{no-state-re\-stric\-tion hypothesis}: the inclusion
of all $\bo\in\R^{d}$ satisfying $\E\cdot\bo$ and $\bu\cdot\bo=1$
in the state space. Note that this is not equivalent to the no-restriction
hypothesis in all cases, for example the noisy bit model in Figure
\ref{fig:NUbit} satisfies the no-state-restriction hypothesis but
not the no-restriction hypothesis.

Continuing the bit example, employing either the no-restriction or
no-state-restriction hypothesis leads to the state space $\mathcal{S}_{B'}$,
the convex hull of the points $(0,1)^{T}$ and $(1,0)^{T}$ pictured
in Figure \ref{fig:effbit}. The pair of state and effect spaces $\mathcal{S}_{B'}$
and $\mathcal{E}_{B'}$ are a transformation of the state and effect
spaces $\mathcal{S}_{B}$ and $\mathcal{E}_{B}$ in Figure \ref{fig:bit}. 

\section{Proofs of Lemmata \ref{thm:Main}, \ref{Lemma WES} and \ref{Lemma: EWE}\label{WESandEWE}}

\prop*
\begin{proof}
Consider a finite set of effects $\E_{1},\E_{2},\ldots,\E_{n}\in\mathcal{E}$
such that $\sum_{j\in J}\E_{j}\in\mathcal{E}$ for any subset $J\subseteq\left\{ 1\ldots,n\right\} $.
First we show that a frame function $v$ must be additive on any such
set, i.e. 
\begin{equation}
v\left(\E_{1}\right)+v\left(\E_{2}\right)+\ldots+v\left(\E_{n}\right)=v\left(\E_{1}+\E_{2}+\ldots\E_{n}\right).\label{eq:additivity}
\end{equation}
We have that the tuples 
\begin{equation}
\left\llbracket \E_{1},\ldots,\E_{n},\bu-\sum_{j=1}^{n}\E_{j}\right\rrbracket \text{, and }\left\llbracket \sum_{j=1}^{n}\E_{j},\bu-\sum_{j=1}^{n}\E_{j}\right\rrbracket ,
\end{equation}
are both observables in the GPT since they satisfy Eqs. (\ref{eq:unit})
and (\ref{eq:coarse}). Hence, by property (V2) of a frame function
we find 
\begin{equation}
\sum_{j=1}^{n}v\left(\E_{j}\right)+v\left(\boldsymbol{u}-\sum_{j=1}^{n}\E_{j}\right)=v\left(\sum_{j=1}^{n}\E_{j}\right)+v\left(\boldsymbol{u}-\sum_{j=1}^{n}\E_{j}\right)=1,
\end{equation}
and Eq. (\ref{eq:additivity}) follows.

The next step is to show the homogeneity of $v$ on $\mathcal{E}$,
i.e. 
\begin{equation}
\alpha v\left(\E\right)=v\left(\alpha\E\right)\quad\text{for all }\E\in\mathcal{\mathcal{E}}\quad\text{and }\alpha\in\left[0,1\right]\,.\label{eq:homogeneity}
\end{equation}
Note that the convexity of $\mathcal{E}$ ensures that rescaling an
effect $\E$ by a factor $\alpha\leq1$ produces another effect: $\alpha\E=\alpha\E+\left(1-\alpha\right)\boldsymbol{0}\in\mathcal{E}$.
For any integer number $n\in\mathbb{N}$, Eq. (\ref{eq:additivity})
implies 
\begin{equation}
v\left(\E\right)=v\left(\frac{n}{n}\E\right)=v\left(\frac{1}{n}\E+\ldots+\frac{1}{n}\E\right)\stackrel{\text{}}{=}nv\left(\frac{1}{n}\E\right)\:;\label{eq:natural}
\end{equation}
then, letting $m\in\mathbb{N}$ with $m\leq n$, Eqs. (\ref{eq:additivity})
and (\ref{eq:natural}) lead to the homogeneity of $v$ over the rationals,
\begin{equation}
v\left(\frac{m}{n}\E\right)=v\left(\frac{1}{n}\E+\ldots+\frac{1}{n}\E\right)\stackrel{\text{}}{=}mv\left(\frac{1}{n}\E\right)\stackrel{}{=}\frac{m}{n}v\left(\E\right).\label{eq:rational}
\end{equation}
Now consider two rational numbers $p,q\in\left[0,1\right]$ with $p\leq q$.
Using property (V1) of a frame function with argument $\left(q-p\right)\E\in\mathcal{E}$
guarantees that $v\left(\left(q-p\right)\E\right)\geq0$. Also we
find by property (V2) of a frame function that 
\begin{equation}
v\left(q\E\right)=v\left(q\E-p\E+p\E\right)\stackrel{\text{}}{=}v\left(\left(q-p\right)\E\right)+v\left(p\E\right)\,.
\end{equation}
Thus, the values of frame functions on multiples of a given effect
respect the ordering induced by the scale factors, 
\begin{equation}
v\left(p\E\right)\leq v\left(q\E\right)\,.\label{eq:orderpreserve}
\end{equation}
Next, let $p_{\mu}$ and $q_{\nu}$ be sequences of rational numbers
in the interval $\left[0,1\right]$ that tend to $\alpha$ from below
and above, respectively. Then we have 
\begin{equation}
p_{\mu}v\left(\E\right)\stackrel{}{=}v\left(p_{\mu}\E\right)\stackrel{}{\leq}v\left(\alpha\E\right)\stackrel{}{\leq}v\left(q_{\nu}\E\right)\stackrel{\text{ }}{=}q_{\nu}v\left(\E\right)\,,
\end{equation}
so that the homogeneity of $v$ claimed in Eq. (\ref{eq:homogeneity})
follows from taking the limit of both sequences.

Thirdly, we construct a well-defined extension of the frame function
$v$ to $\mathcal{E}^{+}$, the positive cone associated with $\mathcal{E}$
(see Definition \ref{def:The-convex-cone}) such that $v\left(\boldsymbol{a}+\boldsymbol{b}\right)=v\left(\boldsymbol{a}\right)+v\left(\boldsymbol{b}\right)$
holds for all $\boldsymbol{a},\boldsymbol{b}\in\mathcal{E}^{+}$.
To do so, consider two effects $\E_{1},\E_{2}\in\mathcal{E}$ which
give rise to the same vector in the positive cone via $\boldsymbol{a}=a_{1}\E_{1}=a_{2}\E_{2}\in\mathcal{E}^{+}$,
with $1<a_{1}<a_{2}$. Then we have 
\begin{equation}
v\left(\E_{2}\right)=v\left(\frac{a_{1}}{a_{2}}\E_{1}\right)\stackrel{\text{}}{=}\frac{a_{1}}{a_{2}}v\left(\E_{1}\right),
\end{equation}
hence $a_{2}v\left(\E_{2}\right)=a_{1}v\left(\E_{1}\right)$, and
we may uniquely define the frame function on arbitrary vectors in
the positive cone by 
\begin{equation}
v\left(\boldsymbol{a}\right):=a_{1}v\left(\E_{1}\right).\label{eq:extensiondef}
\end{equation}
Additivity of the extended frame function is easily seen to hold for
vectors in the positive cone: consider vectors $\boldsymbol{a}=a\E_{a}$
and $\boldsymbol{b}=b\E_{b}$ for $\boldsymbol{e}_{a},\boldsymbol{e}_{b}\in\mathcal{E}$
and $a,b>1$ and let $c=a+b$. Noting that $\left(\boldsymbol{a}+\boldsymbol{b}\right)/c\in\mathcal{E}$
is an effect, we obtain 
\begin{equation}
v\left(\boldsymbol{a}+\boldsymbol{b}\right)\stackrel{}{=}cv\left(\frac{1}{c}\left(\boldsymbol{a}+\boldsymbol{b}\right)\right)\stackrel{}{=}cv\left(\frac{1}{c}\boldsymbol{a}\right)+cv\left(\frac{1}{c}\boldsymbol{b}\right)=v\left(\boldsymbol{a}\right)+v\left(\boldsymbol{b}\right).\label{eq:additive}
\end{equation}

A linear extension of a frame function $v$ to the whole of $\R^{d+1}$
follows from the fact that any $\boldsymbol{c}\in\R^{d+1}$ outside
the positive cone $\mathcal{E}^{+}$ may be decomposed into $\boldsymbol{c}=\boldsymbol{a}-\boldsymbol{b}$
with $\boldsymbol{a},\boldsymbol{b}\in\mathcal{E}^{+}$ by Lemma \ref{Decomposition-1}.
If the decomposition is not unique, $\boldsymbol{c}=\boldsymbol{a}-\boldsymbol{b}=\boldsymbol{a}'-\boldsymbol{b}'$,
we have $\boldsymbol{a}+\boldsymbol{b}'=\boldsymbol{a}'+\boldsymbol{b}$
leading to 
\begin{equation}
v\left(\boldsymbol{a}+\boldsymbol{b}'\right)=v\left(\boldsymbol{a}'+\boldsymbol{b}\right).
\end{equation}
It then follows from Eq. (\ref{eq:additivity}), that $v\left(\boldsymbol{a}\right)+v\left(\boldsymbol{b}'\right)=v\left(\boldsymbol{a}'\right)+v\left(\boldsymbol{b}\right)$
and hence 
\begin{equation}
v\left(\boldsymbol{a}\right)-v\left(\boldsymbol{b}\right)=v\left(\boldsymbol{a}'\right)-v\left(\boldsymbol{b}'\right).
\end{equation}
Therefore we may uniquely define the value of the frame function on
the vector $\boldsymbol{c}$ via 
\begin{equation}
v\left(\boldsymbol{c}\right):=v\left(\boldsymbol{a}\right)-v\left(\boldsymbol{b}\right)\,,\label{eq:fullextension}
\end{equation}
i.e. independently of the decomposition of the vector $\boldsymbol{c}$.

This extension of any frame function $v$ on $\mathcal{E}$ to $\R^{d+1}$
must be linear. First we show additivity: let $\R^{d+1}\ni\boldsymbol{c}_{j}=\boldsymbol{a}_{j}-\boldsymbol{b}_{j}$
for $\boldsymbol{a}_{j},\boldsymbol{b}_{j}\in\mathcal{E}^{+}$, then
\begin{equation}
\begin{aligned}v\left(\boldsymbol{c}_{1}+\boldsymbol{c}_{2}\right) & =v\left(\boldsymbol{a}_{1}-\boldsymbol{b}_{1}+\boldsymbol{a}_{2}-\boldsymbol{b}_{2}\right)\\
 & =v\left(\boldsymbol{a}_{1}+\boldsymbol{a_{2}}-\left(\boldsymbol{b}_{1}+\boldsymbol{b}_{2}\right)\right)\\
 & \stackrel{}{=}v\left(\boldsymbol{a}_{1}+\boldsymbol{a}_{2}\right)-v\left(\boldsymbol{b}_{1}+\boldsymbol{b}_{2}\right)\\
 & \stackrel{}{=}v\left(\boldsymbol{a}_{1}\right)+v\left(\boldsymbol{a}_{2}\right)-v\left(\boldsymbol{b}_{1}\right)-v\left(\boldsymbol{b}_{2}\right)\\
 & \stackrel{}{=}v\left(\boldsymbol{c}_{1}\right)+v\left(\boldsymbol{c}_{2}\right).
\end{aligned}
\end{equation}
Then to show homogeneity let $\R^{d+1}\ni\boldsymbol{c}=\boldsymbol{a}-\boldsymbol{b}$
for $\boldsymbol{a},\boldsymbol{b}\in\mathcal{E}^{+}$, firstly consider
$\gamma\geq0$, in which case we have 
\begin{equation}
\begin{aligned}v\left(\gamma\boldsymbol{c}\right) & =v\left(\gamma\boldsymbol{a}-\gamma\boldsymbol{b}\right)\\
 & \stackrel{}{=}v\left(\gamma\boldsymbol{a}\right)-v\left(\gamma\boldsymbol{b}\right)\\
 & \stackrel{\text{}}{=}\gamma\left(v\left(\boldsymbol{a}\right)-v\left(\boldsymbol{b}\right)\right)\\
 & =\gamma v\left(\boldsymbol{c}\right).
\end{aligned}
\end{equation}
Secondly, consider $\gamma<0$,
\begin{equation}
\begin{aligned}v\left(\gamma\boldsymbol{c}\right) & =v\left(\left(-\gamma\right)\left(-\boldsymbol{c}\right)\right)\\
 & =v\left(\left(-\gamma\right)\left(\boldsymbol{b}-\boldsymbol{a}\right)\right)\\
 & \stackrel{}{=}\gamma\left(v\left(\boldsymbol{a}\right)-v\left(\boldsymbol{b}\right)\right)\\
 & =\gamma v\left(\boldsymbol{c}\right).
\end{aligned}
\end{equation}

Therefore, the extended map admits an expression as 
\begin{equation}
v\left(\boldsymbol{a}\right)=\boldsymbol{a}\cdot\boldsymbol{\omega},
\end{equation}
for some vector $\boldsymbol{\omega}=\sum_{j=1}^{d+1}v\left(\boldsymbol{x}_{j}\right)\boldsymbol{x}_{j}\in\R^{d+1}$,
where $\left\{ \boldsymbol{x}_{1},\ldots\boldsymbol{x}_{d+1}\right\} $
is a basis of $\R^{d+1}$. Finally, requirements (V1) and (V2) on
the behaviour of the frame function $v$ on the effect space $\mathcal{E}$
imply that $\boldsymbol{\omega}\in W\mathcal{\left(E\right)}$ which
concludes the proof.
\end{proof}
\WES*
\begin{proof}
Firstly, by the definitions of the maps $W$ and $E$ in Section \ref{subsec:no-restriction-hypothesis},
we have 
\begin{equation}
W\left(E\left(\mathcal{S}\right)\right)=\left(\mathcal{S}^{*}\cap\left(\boldsymbol{u}-\mathcal{S}^{*}\right)\right)^{*}\cap\boldsymbol{1},
\end{equation}
which, using Lemma \ref{dualconedualofcone}, implies 
\begin{equation}
W\left(E\left(\mathcal{S}\right)\right)=\left(\left(\mathcal{S}^{*}\cap\left(\boldsymbol{u}-\mathcal{S}^{*}\right)\right)^{+}\right)^{*}\cap\boldsymbol{1}.\label{eq:W(E(S))}
\end{equation}

Secondly, we find
\begin{equation}
\left(\mathcal{S}^{*}\cap\left(\boldsymbol{u}-\mathcal{S}^{*}\right)\right)^{+}=\mathcal{S}^{*},\label{eq:dual cone S}
\end{equation}
as shown previously at Eq. (\ref{eq:ESplusSstar}).

Finally, Eqs. (\ref{eq:W(E(S))}) and (\ref{eq:dual cone S}) followed
by an application of Lemma (\ref{Lemma: dual-cone}) and the fact
that $\cS^{+}$ is closed (Lemma \ref{lem:Splusclosed}) give 
\begin{equation}
\begin{aligned}W\left(E\left(\mathcal{S}\right)\right) & =\mathcal{S}^{**}\cap\boldsymbol{1}=\mathcal{S}^{+}\cap\boldsymbol{1}\\
 & =\left\{ \boldsymbol{\omega}\in\R^{d+1}|\boldsymbol{\omega}=x\boldsymbol{\omega}'\text{ for some }\boldsymbol{\omega}'\in\mathcal{S},x\geq0\text{ and }\boldsymbol{\omega}\cdot\boldsymbol{u}=1\right\} ,
\end{aligned}
\label{eq:wesx}
\end{equation}
where the last equality is by the definitions of $\cS^{+}$ and $\boldsymbol{1}$.
Since $\bo'\cdot\bu=1$ for $\bo'\in\cS$ we find the parameter $x$
in Eq. (\ref{eq:wesx}) satisfies $x=x\bo'\cdot\bu=\boldsymbol{\omega}\cdot\boldsymbol{u}=1$
so that
\begin{equation}
\mathcal{S}^{+}\cap\boldsymbol{1}=\mathcal{S}\,,
\end{equation}
as is required for Lemma \ref{Lemma WES} to hold.
\end{proof}
\EWE*
\begin{proof}
Firstly, by the definitions of the maps $W$ and $E$ given in Eqs.
(\ref{eq: E map}) and (\ref{eq: W map}), respectively, we have 
\begin{equation}
E\left(W\left(\mathcal{E}\right)\right)=\left(\mathcal{E}^{*}\cap\boldsymbol{1}\right)^{*}\cap\left(\boldsymbol{u}-\left(\mathcal{E}^{*}\cap\boldsymbol{1}\right)^{*}\right),\label{eq:ewe}
\end{equation}
as well as $\left(\mathcal{E}^{*}\cap\boldsymbol{1}\right)^{*}=\left(\left(\mathcal{E}^{*}\cap\boldsymbol{1}\right)^{+}\right)^{*}$,
by Lemma \ref{dualconedualofcone}.

Secondly, we will show that 
\begin{equation}
\left(\mathcal{E}^{*}\cap\boldsymbol{1}\right)^{+}=\mathcal{E}^{*}.\label{eq:E dual}
\end{equation}
If $\boldsymbol{\omega}\in\mathcal{E}^{*}$ then $\boldsymbol{\omega}\cdot\boldsymbol{u}\geq0$,
which gives 
\begin{equation}
\frac{1}{\boldsymbol{\omega}\cdot\boldsymbol{u}}\boldsymbol{\omega}\in\mathcal{E}^{*}\cap\boldsymbol{1}\,;
\end{equation}
therefore, we conclude that $\boldsymbol{\omega}\in\left(\mathcal{E}^{*}\cap\boldsymbol{1}\right)^{+}$.
Conversely, if $\boldsymbol{\omega}\in\left(\mathcal{E}^{*}\cap\boldsymbol{1}\right)^{+}$,
then $x\boldsymbol{\omega}\in\mathcal{E}^{*}$ for some $x\geq0$,
hence $\boldsymbol{\omega}\in\mathcal{E}^{*}$.

Finally, combining Eqs. (\ref{eq:ewe}) and (\ref{eq:E dual}), we
have 
\begin{equation}
E\left(W\left(\mathcal{E}\right)\right)=\mathcal{E}^{**}\cap\left(\boldsymbol{u}-\mathcal{E}^{**}\right)=\overline{\mathcal{E}^{+}}\cap\left(\boldsymbol{u}-\overline{\mathcal{E}^{+}}\right)\,,
\end{equation}
completing the proof. 
\end{proof}

\section{Proof of Theorem \ref{cor:simGPT}\label{sec:Proof-of-Corollary}}
\begin{proof}
Due to the convexity of the effect space $\mathcal{E}$, we can express
any effect $\E\in\mathcal{E}$ as a convex combination $\E=\sum_{j}p_{j}\E_{j}$,
for some extremal effects $\E_{j}$ and real numbers $p_{j}\in[0,1]$
which sum to one. Thus we may simulate the observable 
\begin{equation}
\mathbb{D}_{\E}=\left\llbracket \E,\boldsymbol{u}-\E\right\rrbracket \label{eq:twooutcome}
\end{equation}
by measuring the observables $\mathbb{D}_{\E_{j}}=\left\llbracket \E_{j},\boldsymbol{u}-\E_{j}\right\rrbracket ,j\in J,$
with probability $p_{j}$. Furthermore, for any effects $\E,\E^{\prime}\in\mathcal{E}$,
we may simulate the observable 
\begin{equation}
\mathbb{T}_{\E,\E^{\prime}}=\left\llbracket \frac{1}{2}\E,\frac{1}{2}\E^{\prime},\boldsymbol{u}-\frac{1}{2}\left(\E+\E^{\prime}\right)\right\rrbracket \label{eq:halves}
\end{equation}
by performing either $\mathbb{T}_{2\E,\mathbf{0}}=\left\llbracket \E,\boldsymbol{0},\boldsymbol{u}-\E\right\rrbracket $
or $\mathbb{T}_{\mathbf{0},2\E^{\prime}}=\left\llbracket \boldsymbol{0},\E^{\prime},\boldsymbol{u}-\E^{\prime}\right\rrbracket $,
with equal probability.

Firstly, applying Definition \ref{simulable frame function} to Eqs.
(\ref{eq:twooutcome}) and (\ref{eq:halves}) with $\E=\E^{\prime}$
gives 
\begin{equation}
v\left(\E\right)+v\left(\boldsymbol{u}-\E\right)=1=v\left(\frac{1}{2}\E\right)+v\left(\frac{1}{2}\E\right)+v\left(\boldsymbol{u}-\E\right),
\end{equation}
and hence 
\begin{equation}
v\left(\E/2\right)=v\left(\E\right)/2.\label{eq:half}
\end{equation}

Secondly, for any effects $\E,\E^{\prime}\in\mathcal{E}$ such that
$\E+\E^{\prime}\in\mathcal{E}$, the observable 
\begin{equation}
\mathbb{D}_{\frac{1}{2}(\E+\E^{\prime})}=\left\llbracket \frac{1}{2}\left(\E+\E^{\prime}\right),\boldsymbol{u}-\frac{1}{2}\left(\E+\E^{\prime}\right)\right\rrbracket ,
\end{equation}
is simulable by Eq. (\ref{eq:twooutcome}). Comparing with Eq. (\ref{eq:halves})
gives 
\begin{equation}
v\left(\frac{1}{2}\E\right)+v\left(\frac{1}{2}\E^{\prime}\right)=v\left(\frac{1}{2}\left(\E+\E^{\prime}\right)\right),
\end{equation}
so that $v\left(\E\right)+v\left(\E^{\prime}\right)=v\left(\E+\E^{\prime}\right)$
follows, using Eq. (\ref{eq:half}). By induction, any simulable frame
function $v$ is a frame function as defined in Definition \ref{A-generalized-probability}.
Thus, by Theorem \ref{thm: NU main}, any simulable frame function
$v$ admits the expression given in Eq. (\ref{eq:gleason sim}). 
\end{proof}

\end{document}